\pdfoutput=1
\documentclass[10pt, doublecolumn]{IEEEtran}

\usepackage{cite}
\usepackage{indentfirst}
\usepackage{graphicx}
\usepackage{changepage}
\usepackage{stfloats}
\usepackage{amsfonts,amssymb}
\usepackage{bm}
\usepackage{algorithm}
\usepackage{algorithmic}
\usepackage{amsmath}
\usepackage{amsmath,amssymb,amsfonts}
\usepackage{times}
\usepackage{url}
\usepackage{cite}
\usepackage{textcomp}
\usepackage{xcolor}
\usepackage{color}
\usepackage{setspace}
\usepackage{enumitem}
\usepackage{float}
\usepackage{diagbox}
\usepackage[T1]{fontenc}
\usepackage[utf8]{inputenc}
\usepackage{authblk}
\usepackage{threeparttable}
\usepackage{booktabs}
\usepackage{footnote}
\usepackage{graphicx}
\usepackage{pifont}
\usepackage{subfigure}
\usepackage{mathrsfs}
\usepackage{bbm}
\usepackage{dsfont}
\usepackage{colortbl}

\allowdisplaybreaks[4]

\newenvironment{proof}{{\indent  \indent \it Proof:}}{\hfill $\blacksquare$}
\begin{document}
\title{{Network-level ISAC: An Analytical Study of Antenna Topologies Ranging from Massive to Cell-Free MIMO}}

\author{
	Kaitao Meng, \textit{Member, IEEE}, Kawon Han, \textit{Member, IEEE}, Christos Masouros, \textit{Fellow, IEEE}, and Lajos Hanzo, \textit{Life Fellow, IEEE} 
	\thanks{Preliminary versions of this paper were presented at the Conference IEEE WCNC 2025  \cite{Meng2025WCNC}.}
	\thanks{K. Meng, K. Han, and C. Masouros are with the Department of Electronic and Electrical Engineering, University College London, London, UK (emails: \{kaitao.meng, kawon.han, c.masouros\}@ucl.ac.uk). L. Hanzo is with School of Electronics and Computer Science, University of Southampton, SO17 1BJ Southampton, UK (email: lh@ecs.soton.ac.uk) 
}}

\maketitle

%%%%%%%%%%%%%%%%%%%%%%%%%%%%%%%%%%%%%%%%%%%%%%%%%%%%%%%%%%%%%%%%%%%%%%%%%%%%%%%%%
%angle-of-arrival (AOA) time-of-flight (TOF)
\begin{abstract}
A cooperative architecture is proposed for integrated sensing and communication (ISAC) networks, incorporating coordinated multi-point (CoMP) transmission along with multi-static sensing. We investigate how the allocation of antennas-to-base stations (BSs) affects cooperative sensing and cooperative communication performance. More explicitly, we balance the benefits of geographically concentrated antennas in the massive multiple input multiple output (MIMO) fashion, which enhance beamforming and coherent processing, against those of geographically distributed antennas towards cell-free transmission, which improve diversity and reduce service distances. Regarding sensing performance, we investigate three localization methods: angle-of-arrival (AOA)-based, time-of-flight (TOF)-based, and a hybrid approach combining both AOA and TOF measurements, for critically appraising their effects on ISAC network performance. Our analysis shows that in networks having \( N \) ISAC nodes following a Poisson point process, the localization accuracy of TOF-based methods follows a \( \ln^2 N \) scaling law (explicitly, the Cramér-Rao lower bound (CRLB) reduces with \( \ln^2 N \)). The AOA-based methods follow a \( \ln N \) scaling law, while the hybrid methods scale as \( a\ln^2 N + b\ln N \), where \( a \) and \( b \) represent parameters related to TOF and AOA measurements, respectively. The difference between these scaling laws arises from the distinct ways in which measurement results are converted into the target location. Specifically, when converting AOA measurements to the target location, the localization error introduced during this conversion is inversely proportional to the distance between the BS and the target, leading to a more significant reduction in accuracy as the number of transceivers increases. In contrast, TOF-based localization avoids such distance-dependent errors in the conversion process.
In terms of communication performance, we derive a tractable expression for the communication data rate, considering various cooperative region sizes and antenna-to-BS allocation strategy. It is proved that higher path loss exponents favor distributed antenna allocation to reduce access distances, while lower exponents favor centralized antenna allocation to maximize beamforming gain. Simulations confirm that cooperative transmission and sensing in ISAC networks can effectively improve non-cooperative sensing and communication performance The proposed cooperative scheme shows superior performance improvement compared to centralized or distributed antenna allocation strategies.
\end{abstract}   


%%%%%%%%%%%%%%%%%%%%%%%%%%%%%%%%%%%%%%%%%%%%%%%%%%%%%%%%%%%%%%%%%%%%%%%%%%%%%%%%
\newtheorem{thm}{\bf Lemma}
\newtheorem{remark}{\bf Remark}
\newtheorem{Pro}{\bf Proposition}
\newtheorem{theorem}{\bf Theorem}
\newtheorem{Assum}{\bf Assumption}
\newtheorem{Cor}{\bf Corollary}

\section{Introduction}

Given the increasing challenges of spectrum scarcity and potential interference between separate sensing and communication (S\&C) systems, integrated sensing and communication (ISAC) technologies have garnered substantial academic and industrial interest \cite{Meng2025WCNC, Zhang2021OverviewSignal, Liu2022SurveyFundamental, Meng2023SensingAssisted}. ISAC is recognized for its ability to leverage unified infrastructure and waveforms to simultaneously transmit information and receive echoes, thereby significantly enhancing the spectrum-, cost-, and energy-efficiency \cite{Cui2021Integrating}. Recently, the international telecommunication union (ITU) identified ISAC as one of the six key usage scenarios for the forthcoming sixth-generation (6G) networks. While current research is primarily focused on link-level and system-level optimization, such as waveform design and resource management within individual base stations (BSs) \cite{Ouyang2022Performance, Liu2022Integrated, Meng2023Throughput, Liu2023DistributedUnsupervised, Hua20243DMultiTarget}, the broader opportunities of network-level ISAC, particularly multi-cell S\&C cooperation, have not been widely explored \cite{meng2024integrated}.

Network-level ISAC presents distinct advantages over conventional single-cell ISAC, including expanded coverage, enhanced service quality, more flexible performance tradeoffs, and the ability to gather richer target information \cite{Li2023Towards, Shin2017CoordinatedBeamforming, DeSaintMoulin2023Cooperative}. Specifically, with the exploitation of the target-reflected signals over both the monostatic links (BS-to-target-to-originated BS links) and the multistatic links (BS-to-target-to-other BSs links), the sensing capabilities of ISAC can be maximized through multi-static sensing configurations formed by several cooperative BSs. Additionally, advanced coordinated multi-point (CoMP) transmission and reception techniques can be employed for mitigating inter-cell interference, enhancing communication performance by connecting a single user to multiple BSs for reliable connection and improved throughput \cite{Hosseini2016Stochastic}. The strategic integration of cooperative sensing and communication techniques within ISAC networks offers substantial potential to enhance and dynamically balance the S\&C performance. Some early studies have explored network-level trade-offs between sensing and communication \cite{Meng2024CooperativeTWC, salem2022rethinking}, focusing on aspects such as waveform design and optimizing cooperative cluster sizes. For instance, in \cite{meng2023network}, coordinated beamforming was employed to mitigate interference in ISAC networks, providing valuable insights into the optimal allocation of spatial resources. However, most existing research neglects to explore the rationale behind the specific antenna configurations of ISAC networks \cite{Meng2024CooperativeTWC, salem2022rethinking, meng2023network}, often assuming fixed setups without examining how the antenna-to-BS allocation affects the overall network performance.

In ISAC networks, optimal antenna-to-BS allocation, represented by the number of antennas per site, plays a critical role in maximizing the cooperative gains for both sensing and communication, since these two functions have fundamentally different requirements for their antenna configurations.
Typically, the antenna-to-BS allocation strategies fall into two main categories, namely centralized and distributed configurations. Centralized multiple input multiple output (MIMO) systems simplify deployment and reduce costs by concentrating antennas in a single location within the service region, such as conventional cellular MIMO networks\cite{Overview2017Rappaport}. However, this approach is prone to high spatial channel correlation, particularly at millimeter-wave frequencies, leading to significant performance erosion. By contrast, distributed MIMO configurations, where the antennas are dispersed across various locations, can mitigate channel correlation and enhance system performance by providing improved spatial diversity and reducing targets/users access distances; a prominent example is the cell-free MIMO system \cite{Ngo2017CellFree}. The primary drawback of distributed MIMO systems is the challenge of achieving and maintaining precise synchronization across all antennas, which is crucial for improving the coherent processing gain of ISAC networks. These synchronization requirements can offset the performance gains and present significant obstacles to fully realizing the potential of distributed ISAC networks due to their high overhead and design complexity. Upon evolving from communication-only to ISAC networks, the traditional antenna configuration strategies tailored for communication may not directly apply to the ISAC paradigm \cite{Koyuncu2018Performance}, especially for various target information measurements, such as angle, distance, and velocity. Consequently, this shift requires innovative ISAC cooperation approaches, focusing on optimizing antenna resource allocation to accommodate the distinct demands of sensing.

Building on the previous discussions, we propose a cooperative ISAC scheme that combines the benefits of both centralized and distributed antenna allocation strategies, for carefully allocating antenna resources and for balancing the spatial diversity with the number of antennas per site. Specifically, concentrating more antennas at selected locations enhances beamforming gain and coherent processing but requires a reduction in BS density and increases the average service distance. Conversely, distributing antennas across multiple locations improves geometric diversity, allowing for enhanced sensing and communication services over shorter distances.
In the literature, some related work proposed to optimize antenna configurations based on the specific location of users/targets \cite{Cai2015Deployment, Radmard2014Antennaplacement, Gorji2014Optimal}, thereby increasing system throughput. However, in practice, user locations, target positions, and channel conditions are unpredictable, requiring antenna-to-BS allocation strategies that account for the randomness in user, target, and BS locations, as well as channel fading variability. These factors make it challenging to precisely characterize the relationship between antenna-to-BS allocation and the resulting sensing as well as communication performance.

\begin{table*}[t]
	\centering
\footnotesize
\caption{Comparison of Related Works on ISAC Design Considerations}

	\begin{tabular}{|l|c|c|c|c|c|c|}
		\hline
		{Existing works} & {[25], [26]} & {[27]} & [28] & [24]  & [23] &  {\textbf{This work}} \\ \hline
		{AOA Measurement} &  &  & $\checkmark$ &  & & $\checkmark$ \\  \hline
		{TOF Measurement} &  & $\checkmark$ &  &  & & $\checkmark$\\  \hline
		{Antenna Allocation} &  &  &  & $\checkmark$ & $\checkmark$ & $\checkmark$ \\  \hline
		{Cooperative Sensing} &  & $\checkmark$  & $\checkmark$ & $\checkmark$ & & $\checkmark$ \\  \hline
		{Cooperative Communication} & $\checkmark$  &  &  &   & $\checkmark$ & \textbf{$\checkmark$} \\
		\hline 
		{Antenna Topologies Analysis} &  &  &  &   &  & \textbf{$\checkmark$} \\
		\hline 
		{Scaling Law Derivation} &  &  &  &   &  & \textbf{$\checkmark$} \\
		\hline 
\end{tabular}
\label{tab:comparison}
\end{table*}

To handle the above issue, stochastic geometry offers a powerful mathematical framework for analyzing multi-cell wireless sensing and communication networks \cite{demir2021foundations}. For instance, the framework proposed in \cite{Andrews2011TractableApproach} provides insights into average data rate and coverage probability in multi-cell communication networks. In \cite{wiame2023joint}, the statistics of data rate and incident power density in user-centric cell-free networks are analyzed using stochastic geometry, deriving useful performance metrics and joint distribution bounds. As a further advance, \cite{Schloemann2016Toward} examines a sensing metric based on the signal-to-interference-plus-noise ratio (SINR) to establish the relationship between sensing accuracy and the number of BSs transmitting signals with sufficient power for effective localization. Also, in \cite{Moulin2024JointCoverage}, the authors introduced novel performance metrics for analyzing full-duplex ISAC systems, capturing joint statistics of both functions and optimizing network parameters while considering mutual interference cancellation. In \cite{Ram2022Frontiers}, the BS serves as a dual functional transmitter that supports both S\&C functionalities in a time division manner, where during the search interval, the radar scans the entire angular search space to find the maximum number of mobile users.
Furthermore, it is noteworthy that the assessment of the sensing performance relying on metrics like the SINR or mutual information overlooks the position geometry of the cooperating BSs \cite{meng2023network, Schloemann2016Toward, Lone2018Statistical}. It is noted that the analysis of sensing performance using other metrics from estimation theory, such as the Cramér-Rao lower bound (CRLB) that accounts for the node geometry \cite{Eldar2006Uniformly,  Lu2024Integrated}, is essential as it effectively links the time-of-flight (TOF) and angle-of-arrival (AOA) \cite{He2022PerformanceAOA} measurements with the estimated location. However, describing network performance based on the CRLB is challenging, as it involves complex operations, such as expectations over matrix inversions.

In this treatise, we propose a cooperative ISAC scheme, as shown in Fig.~\ref{figure1}, where multiple BSs within the cooperative communication region cooperatively transmit the same information data to the served user, while another set of BSs within the cooperative sensing region collaborate with the objective of offering localization services for each target. By strategically integrating CoMP-based joint transmission with multi-static sensing, this scheme aims for striking a tradeoff between sensing and communication performance at the network level.
In this work, we investigate three different target localization methods: AOA-based, TOF-based, and a hybrid of AOA and TOF based localization, for comprehensively assessing the impact of antenna-to-BS allocation on cooperative sensing and communication performance in ISAC networks. Based on this, we reveal the performance of different antenna configuration strategies and their corresponding scaling laws, as summarized in Table \ref{Table2}.
In contrast to the most relevant studies without antenna-to-BS allocation design \cite{Meng2024CooperativeTWC, salem2022rethinking}, in this work, both the number of antennas per site and the power allocation of S\&C are optimized for further improving the whole network's performance and achieving a more flexible tradeoff between the S\&C performance at the network level. 
{\textbf{The main contributions of this paper are summarized as follows:}}
\begin{itemize}[leftmargin=*]
	\item {\textit{We propose a cooperative ISAC network that integrates multi-static sensing with CoMP data transmission.}} Our approach includes a localization method that exploits AOA and TOF measurements. By leveraging our model and stochastic geometry tools, we offer analytical insights into the performance of both sensing and CoMP, highlighting key dependencies related to antenna-to-BS allocation in ISAC networks.
	\item By analyzing the scaling laws of network localization schemes, we find that, given that $N$ ISAC BSs are employed, TOF-based methods follow a scaling law of \( \ln^2 N \), where the CRLB reduces with \( \ln^2 N \). By contrast, the AOA-based methods follow a scaling law of \( \ln N \), and hybrid methods follow a scaling law of \( a\ln^2 N +  b\ln N\), where $a$ and $b$ denote the parameters related TOF measurements and AOA measurements, respectively. The primary difference stems from the fact that converting AOA measurements to the target position introduces additional distance-related variables, resulting in a smaller scaling law for AOA-based methods. However, a hybrid localization approach that combines both TOF and AOA measurements can significantly enhance localization performance, especially when the number of ISAC BSs is relatively small. This is achieved by fully leveraging the strengths of localization methods based on AOA and TOF measurements.
	\item We derive the effective channel gain and the Laplace transform of both the useful signals and inter-cell interference by utilizing the moment-generating functions. Based on this, we establish a tractable expression for the communication data rate for various antenna-to-BS allocation strategies and cooperative region sizes. Additionally, we determine the optimal antenna-to-BS allocation strategy of special cases, showing that larger path loss exponents favor distributed allocation for reducing access distances, while smaller exponents favor centralized allocation to maximize beamforming gain. 
	\item A performance boundary optimization problem is studied for ISAC networks, and
	we verify that cooperative transmission and sensing in ISAC networks can effectively improve the S\&C gain and strike a more flexible tradeoff between the S\&C performance. Moreover, it is revealed that when provided with more antenna resource blocks, the proposed cooperative scheme exhibits a more substantial performance improvement than fully centralized or fully distributed allocation schemes.
\end{itemize}

Notation: $B(a,b,c) = \int_0^a t^{b-1} (1-t)^{c-1}dt$ is the lower incomplete Beta function, and $\bar B(a,b,c) = \int_a^1 t^{b-1} (1-t)^{c-1}dt$ is the upper incomplete Beta function. Lower-case letters in bold font will denote deterministic vectors. For instance, $X$ and ${\bf{X}}$ denote a one-dimensional (scalar) random variable and a random vector (containing more than one element), respectively. Similarly, $x$ and ${\bf{x}}$ denote scalar and deterministic vectorial values, respectively. ${\rm{E}}_{x}[\cdot]$ represents statistical expectation over the distribution of $x$, and $[\cdot]$ represents a variable set. ${\cal{C}}(0,r)$ denotes the circle region with center at the origin and radius $r$. The frequently used parameters and variables are give in Table \ref{Notation}.

\section{System Model}

\begin{figure}[t]
	\centering
	\footnotesize
	\includegraphics[width=8.2cm]{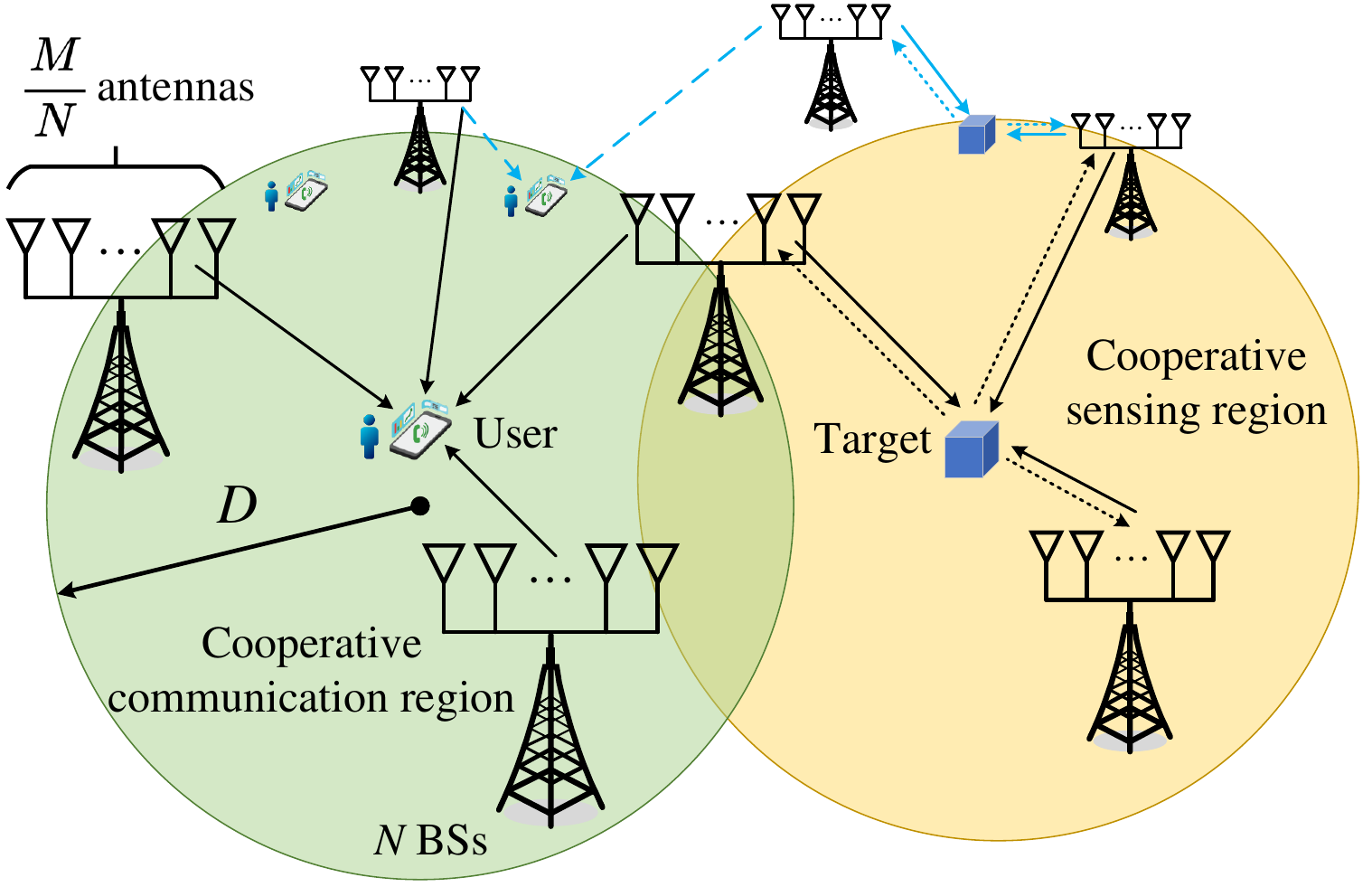}
	\vspace{0mm}
	\caption{Illustration of antenna-to-BS allocation in cooperative ISAC networks with optimized BS density (Blue line refers to other time-frequency resources).}
	\label{figure1}
\end{figure}

\begin{table}[]
	\centering
	\footnotesize
	\caption{Localization Method comparisons} 
	\label{Table2}
	\begin{tabular}{
			>{\columncolor[HTML]{CBCEFB}}l 
			>{\columncolor[HTML]{EFEFEF}}l 
			>{\columncolor[HTML]{CBCEFB}}l 
			>{\columncolor[HTML]{EFEFEF}}l }
		\hline
		{\color[HTML]{333333} \textbf{\begin{tabular}[c]{@{}l@{}} Localization\\ Method\end{tabular}}} & {\color[HTML]{333333} \textbf{\begin{tabular}[c]{@{}l@{}}Time \\ Synchronization\end{tabular}}} & {\color[HTML]{333333} \textbf{\begin{tabular}[c]{@{}l@{}}Accuracy \\ factors\end{tabular}}}                                               & {\color[HTML]{333333} \textbf{\textbf{\begin{tabular}[c]{@{}l@{}}CRLB  \\ Scaling law$^\star$\end{tabular}}}} \\  \hline
		{\color[HTML]{333333} \begin{tabular}[c]{@{}l@{}}Angle\\ measurement\end{tabular}}            & {\color[HTML]{333333} No}                                                                      & {\color[HTML]{333333} \begin{tabular}[c]{@{}l@{}}Array aperture size \\ 
				Array orientation\end{tabular}} & {\color[HTML]{333333} $\frac{1}{{\ln N}}$} \\
		\hline 
		{\color[HTML]{333333} \begin{tabular}[c]{@{}l@{}}Range\\ measurement\end{tabular}}            & {\color[HTML]{333333} Yes}                                                                     & {\color[HTML]{333333} Bandwidth}                                                    & { \color[HTML]{333333} $\frac{1}{\ln^2 N}$}        \\
		\hline
		{\color[HTML]{333333} \begin{tabular}[c]{@{}l@{}}Hybrid\\ measurement\end{tabular}}            & {\color[HTML]{333333} Yes}                                                                     & {\color[HTML]{333333} All above}                                                   & { \color[HTML]{333333} $\frac{1}{a \ln^2 N + b \ln N}$}        \\ \hline
		\noalign{\vskip 2mm}   
		\hline
	\end{tabular}
	\begin{tablenotes}
		\footnotesize
		%\item $\heartsuit$ TS represents whether the localization requires time synchronization.
		\item $^\star$ $N$ denotes the number of BSs in the cooperative sensing cluster under PPP distribution. $a$ and $b$ denote the parameters of TOF measurements and AOA measurements, respectively.
	\end{tablenotes}
\end{table}

\subsection{Network Model}

As shown in Fig. \ref{figure1}, BSs inside the cooperative region, defined by the circle centered at the target/user with radius $D$, form a cluster for cooperative sensing and communication. Specifically, for communication cooperation, each user is served by multiple BSs within this region, where the BSs transmit identical signals to the user by forming a CoMP cluster, thereby enhancing the received signal power through constructive signal superposition. Similarly, for target localization, BSs within the target-defined cooperation region collaborate as a distributed multi-static MIMO radar system, employing a code-division multiplexing scheme to maintain orthogonality among transmitted waveforms, ensuring accurate target localization.
Given the significant cost and complexity associated with phase-level synchronization required for coherent processing, our design prioritizes a more practical non-coherent approach. In particular, BS clusters adopt non-coherent joint transmission for communication and non-coherent MIMO radar processing for sensing. This choice strikes a balance between performance and feasibility, enabling effective cooperative operations while mitigating the synchronization challenges typically encountered in large-scale distributed deployments.

\begin{table}[t]
	\small
	\caption{Important notations and symbols.} 
	\label{Notation}	
	\centering
	\begin{tabular}{ll}% 
			\hline
			$\rm{\textbf{Notation}}$  & $\rm{\textbf{ Physical meaning}}$ \\
			\hline 
			${\bf{F}}_{\rm{A}}$        &  FIM for AOA-based localization \\
			\hline 
			${\bf{F}}_{\rm{R}}$        &  FIM for TOF-based localization \\	
			\hline 
			${\bf{F}}_{\rm{H}}$        &  FIM for hybrid AOA and TOF localization \\
			\hline 
			$\tilde {\bf{F}}_{\rm{A}}$ &  FIM only with geometry factors for AOA localization \\	
			\hline 
			$\tilde {\bf{F}}_{\rm{R}}$ &  FIM only with geometry factors for TOF localization \\		
			\hline 
			$\tilde {\bf{F}}_{\rm{H}}$ &  FIM only with geometry factors for hybrid localization \\
			\hline  
			${\mathbf{W}}_i$               &  Beamforming matrix of BS $i$  \\
			\hline  
			${\mathbf{s}}_i$              &  Combined sensing and communication signals of BS $i$ \\
			\hline  
			${\mathbf{x}}_i$              &  transmitted signals after precoding of BS $i$ \\
			\hline  
			${{s}}^c_i$              &  Communication signals of BS $i$ \\
			\hline  
			${{s}}^s_i$              &  Sensing signals of BS $i$ \\
			\hline 
			$p^s$                     &  Transmit power of sensing signal \\
			\hline  
			$p^c$                     &  Transmit power of communication signal \\
			\hline  
			$M_{\rm{t}}$, $M_{\rm{r}}$ &  Numbers of transmit and receive antennas per BS \\
			\hline  
			$\theta_i$ &  Angle of bearing for i-th BS to the target\\
			\hline  
			$\alpha$                   &  Pathloss factor of communication channel\\
			\hline  
			$\beta$                   &  Pathloss factor of communication channel\\
			\hline  
			$R_c$                     &  Average data rate  \\
			\hline  	
	\end{tabular}
\end{table}

In this study, we explore the optimal strategy for allocating antennas to BSs in cooperative ISAC networks, given the constraints on antenna resources. Allocating multiple antennas in a single array can enhance beamforming gain and coherent processing gain for sensing and communication. Conversely, the allocation of antennas in a distributed manner can enhance macro-multiplexing gain by enabling S\&C service delivery at a closer distance and can improve target localization accuracy leveraging the spatial diversity of distributed MIMO radar. Thus, we design a novel antenna-to-BS allocation scheme to strike a fundamental tradeoff between coherent processing gain, macro-multiplexing gain, and geometric diversity gain in cooperative ISAC networks. %, and can improves target localization accuracy leveraging the spatial diversity of distributed MIMO radar for sensing

Given a fixed total number of antennas, we aim for maximizing the cooperative S\&C performance by optimizing the BS density and the number of antennas allocated at each BS, as shown in Fig. \ref{figure1}. Specifically, we define antenna density and BS density as the average number of antennas per km$^2$ and the average number of BS per km$^2$, respectively. Then, given the transmit antenna density $\lambda_t$ and receive antenna density $\lambda_r$, assuming each BS is deployed with a uniform linear array, the BS density to be optimized is denoted by $\lambda_{b} = \frac{\lambda_{t}}{M_t} = \frac{\lambda_{r}}{M_r}$, where \(M_t\) and \(M_r\) represent the number of transmit and receive antennas per BS. With the fixed antenna density, increasing the number of antennas per BS will reduce the overall BS density. For example, when $M_t = M_r = 1$, the antennas are allocated in a distributed manner. By contrast, given an area $|A|$, when all antennas in this area are deployed at a certain location, we have a centralized allocation. 
It is assumed that the locations of BSs follow a homogeneous Poisson point process (PPP) in a two-dimensional (2D) space, denoted by $\Phi_b$, where PPP offers analytical tractability and realistic modeling of spatial randomness in cellular networks. Here, $\Phi_b = \{ {\bf{d}}_i = [x_i,y_i]^T \in \mathbb{R}^2, \forall i \in \mathbb{N}^+ \}$, where $\mathbf{d}_i$ respectively represents the location of BS $i$.

%\begin{figure}[t]
%	\centering
%	\footnotesize
%	\includegraphics[width=8.4cm]{figure2.pdf}
%	\vspace{-3mm}
%	\caption{Illustration of various antenna topologies.}
%	\label{figure2}
%\end{figure}

Each BS designs the transmit precoding (TPC) for sending the information signal $s_i^c(t)$ to the served user, together with a dedicated radar signal $s^{s}_i(t)$ for the detected target, where  the variable $t$ represents time instant.\footnote{Service requests are sent by users and targets to nearby BSs. When a BS is connected to multiple users and targets, it is assumed that they are scheduled to different orthogonal time/frequency resource blocks. This ensures that each BS serves at most one user and one target in a given time/frequency slot. While this simplified scheme is not optimal, it facilitates the derivation of tractable expressions.} This is consistent with the assumptions in \cite{Liu2020JointTransmit, Hua2023Optimal}, $\mathrm{E}[s^s_i(t) (s^c_i)^H(t)] = 0$. In the following discussion, we omit $(t)$ from the S\&C signal notation for simplicity.  Upon letting ${\mathbf{s}_i=\left[s^c_i, s^s_i\right]^T}$, we have $\mathrm{E}\left[\mathbf{s}_i \mathbf{s}_i^H\right]=\mathbf{I}_2$. Then, the signal transmitted by the $i$th BS is given by\footnote{If the BS only receives a request for communication services, the sensing component is omitted, and the transmitted signal simplifies to
$
{\bf{x}}_i = {\bf{w}}_i^c {\bf{s}}^{c}_i.
$
Similarly, if the BS only receives a request for sensing services, the transmitted signal becomes
$
{\bf{x}}_i = {\bf{w}}_i^s {\bf{s}}^{s}_i.
$}
\vspace{0mm}
\begin{equation}\label{TrasmitSignals}
	{\mathbf{x}}_i = \mathbf{W}_i \mathbf{s}_i =  {\bf{w}}^c_i s^c_i +   {\bf{w}}^s_i s^s_i,
	\vspace{0mm}
\end{equation}
where ${\bf{w}}^c_i$ and ${\bf{w}}^s_i \in {\mathbb {C}}^{M_{\mathrm{t}} \times 1}$ are normalized transmit beamforming vectors, i.e., $\|{\bf{w}}^c_i \|^2 = p^c$ and $\|{\bf{w}}^s_i\|^2 = p^s$. Here, $p^s$ and $p^c$ respectively represent the transmit power of the S\&C signals, and $\mathbf{W}_i=\left[\mathbf{w}^c_i, \mathbf{w}^s_i\right] \in {\mathbb{C}}^{M_{\mathrm{t}} \times 2}$ is the TPC matrix of the BS at $\mathbf{d}_i$. To avoid the interference between S\&C, we adopt zero-forcing (ZF) beamforming for the sake of making the analysis tractable. Then, the beamforming matrix of the serving BS $i$ is given by 
\vspace{0mm}
\begin{equation}\label{TransmitBeamforming}
	{\bf{W}}_i = {\tilde {\mathbf{W}}}_i  \left(\sqrt{\operatorname{diag}\left({\tilde {\mathbf{W}}}_i^H {\tilde {\mathbf{W}}}_i\right)}\right)^{-1},
	\vspace{0mm}
\end{equation}
where ${\tilde {\mathbf{W}}}_i = {{\bf{H}}^H_i}{\left( {\bf{H}}_i  {\bf{H}}_i^H \right)^{-1}}$ and $\mathbf{H}_i = 
\begin{bmatrix}
	\mathbf{h}_{i,c}^H \\[0.5ex]
	\mathbf{a}^H(\theta_i)
\end{bmatrix}
\;\in\;\mathbb{C}^{2\times M_t}$. Here, ${\bf{h}}^H_{i,c} \in \mathbb{C}^{1 \times M_{\mathrm{t}}}$ denotes the communication channel 
spanning from BS $i$ to the typical user, and ${{\bf{a}}^H}(\theta_i ) \in \mathbb{C}^{1 \times M_{\mathrm{t}}}$ represents the sensing channel impinging from BS $i$ to the typical target. We have $p^s + p^c = 1$ with normalized transmit power. Upon using ZF beamforming, inter-cell communication interference is effectively mitigated because all BSs within the cooperative cluster provide service to the same user, while avoiding the use of sensing beams directed at this served communication user.%\footnote{ If BSs collaborative serves multiple users and targets simultaneously  within one time-frequency block, a similar approach can be applied to analyze channel gain and large-scale fading.}

\subsection{Cooperative Sensing Model}
\label{CooperativeSensing}
We aim for exploring the optimal antenna-to-BS allocation method by examining the scaling laws of target localization techniques that rely on AOA measurements, TOF measurements, and a combination of both, respectively.
The location of a typical target is denoted as ${\bm{\psi}}_t = [x_t, y_t]^T$.\footnote{According to Slivnyak's theorem \cite{Andrews2011TractableApproach, mukherjee2014analytical}, the typical target is assumed to be located at the origin. Its performance is evaluated to determine the average performance of all targets across the network, using the probability distribution function of the distances from the BSs to this origin. Similarly, the average communication performance of a typical user located at the origin is assessed by analyzing the distance from the typical user to the BSs based on their location distribution.} Assuming unbiased estimations, the CRLB serves as a benchmark for theoretical localization accuracy in terms of the mean squared error (MSE), which can be expressed as
\vspace{0mm}
\begin{equation}	
	{\rm{var}}\{\hat{{\bm{\psi}}}_t\} = {\rm{E}} \{| \hat{{\bm{\psi}}}_t - {{\bm{\psi}}}_t |^2\} \ge \mathrm{CRLB},
	\vspace{0mm}
\end{equation}
where $\hat{{\bm{\psi}}}_t=\left[\hat{x}_t, \hat{y}_t\right]^T$ represents the estimated location of the typical target. The typical target is collaboratively sensed by $N$ BSs. Let us assume that the transmitted radar signals $\{s_i^s\}_{i=1}^N$ of the BSs in the cooperative sensing cluster are approximately orthogonal for any time delay of interest \cite{li2008mimo}. %\footnote{Due to the substantial amount of original echo data, the echo signals received at each BS can be pre-processed (e.g., coherent operations) in a distributed manner, and then send the processing results to the central unit through backhaul links, which can significantly reduce the backhaul overhead.}  
The base-band equivalent of the impinging signal at receiver $j$ is represented as %Due to the adoption of ZF beamforming, 
\begin{equation}\label{SensingChannel}
	\begin{aligned}
		{{\bf{y}}_{j}}(t) =& \sum\nolimits_{i = 1}^N \sigma \underbrace {{{d}_j^{ - \frac{\beta}{2} }}{\bf{b}}\left( {{\theta _j}} \right){{d}_i^{ - \frac{\beta}{2} }}{{\bf{a}}^H}\left( {{\theta _i}} \right)}_{{\text{target channel}}}{{\bf{W}}_i} {{\bm{s}}_i}(t - \tau_{i,j}) \\
		&+ \underbrace {\sum\nolimits_{i \in {\Phi_b}} { {{\bf{H}}_{i,j}} \mathbf{W}_i \mathbf{s}_i(t - \tilde{\tau}_{i,j}) }}_{{\text{inter-cluster interference}}} +  {\bf{n}}_l(t),
	\end{aligned}
\end{equation}
where $d_i = {\left\| {{{\bf{d}}_i}} \right\|}$ denotes the distance from BS $i$ to the origin, $\beta \ge 2$ is the pathloss exponent between the serving BS and the typical target.\footnote{In general, sensing relies exclusively on line-of-sight (LoS) channels, i.e., the reflection signal from NLoS channel is useless for sensing \cite{Liu2022SurveyFundamental}; while communication is modeled by circularly symmetric Gaussian distributions that capture both LoS and NLoS conditions, leading to fundamentally different channel coefficients.} Furthermore, $\sigma$ denotes the radar cross section (RCS), $\tau _{i,j}$ is the propagation delay of the link spanning from BS $i$ to the typical target and then to BS $j$, and $\tilde \tau _{i,j}$ denotes the propagation delay of the link from BS $i$ to BS $j$. In (\ref{SensingChannel}), ${\bf{H}}_{i,j}$ denotes the channel from BS $i$ to BS $j$. Finally, the term ${\bf{n}}_l(t)$ is the additive complex Gaussian noise having zero mean and covariance matrix $\sigma_s^2 {\bf{I}}_{M_{\rm{r}}}$. In (\ref{SensingChannel}), we have ${{\bf{a}}^H}(\theta_i ) = [1, \cdots, e^{ {j \pi(M_{\mathrm{t}}-1)  \cos(\theta_i) }}]$, and ${\bf{b}}(\theta_j ) = [1, \cdots, e^{ {j \pi(M_{\mathrm{r}}-1)  \cos(\theta_j) }}]^T$, where $\theta_i$ denotes the angle of bearing for the $i$-th BS to the target with respect to the horizontal axis. The cooperative sensing model neglects interference from other targets’ echoes, as distinct spatial locations, angles, and propagation delays allow advanced filtering techniques, like matched and adaptive filtering, to effectively mitigate it, consistent with related works \cite{sadeghi2021target}.

%We assume that the target is stationary during the short signal processing interval, and thus the RCS distribution obeys a Swerling-I model \cite{skolnik1980introduction}. 

\subsubsection{AOA Measurement based Localization}
For the AOA estimate at the receiver $j$, the covariance matrix can be given by
	$
	\mathbf{R}_y = \mathbb{E}\left\{\mathbf{y}_j(t)\mathbf{y}_j^H(t)\right\},
	$
	and we perform an eigen-decomposition to separate the signal and noise subspaces. The MUSIC algorithm can be applied to this decomposition to generate a pseudospectrum, where the peak corresponds to the estimated AOA \cite{Stoica1989MUSIC}. 
		By measuring the AOAs of each monostatic link and bi-static link, the target location can be estimated by maximum likelihood estimation (MLE) \cite{li1993maximum}. For the AOA measurement of the bi-static link from the $j$th BS to the target and then to the $i$th BS, we have
\begin{equation}\label{AngleEstimation}
	\hat {\theta}_{i,j} = \tan ^{-1} \frac{{y}_t-{y_i}}{{x}_t-{x_i}}+n^a_{i,j}.
\end{equation}
In (\ref{AngleEstimation}), $n^a_{i,j}$ denotes the AOA measurement error, and $n^a_{i,j} \sim \mathcal{N}\left(0, \rho_{i,j}^2\right)$, where $\rho_{i,j}^2 = \frac{6}{\pi^2  \cos ^2 \theta_i  M_r\left(M_r^2-1 \right) G_t \gamma_{i,j}}$ \cite{richards2005fundamentals} and $\gamma_{i,j} = \frac{\sigma p^s \gamma_0}{d_i^{\beta}d_j^{\beta}}$.  Here, $G_t$ is the transmit beamforming gain, and $\gamma_0$ represents the channel power at the reference distance of 1 m. Then, we transform $N^2$ AOA measurement links into the target location. The Jacobian matrices of the $N$ BS measurement errors, evaluated at the true target position ${\bm{\psi}}_t = [x_t, y_t]^T$, indicate
\begin{equation}
	\mathbf{J}_{\mathrm{A}}=\left[\begin{array}{cc}
		\frac{\partial \hat{\theta}_1}{\partial x_t} & \frac{\partial \hat{\theta}_1}{\partial y_t} \\
		\vdots & \vdots \\
		\frac{\partial \hat{\theta}_N}{\partial x_t} & \frac{\partial \hat{\theta}_N}{\partial y_t}
	\end{array}\right]=\left[\begin{array}{cc}
		\frac{-\sin \theta_1}{d_1} & \frac{\cos \theta_1}{d_1} \\
		\vdots & \vdots \\
		\frac{-\sin \theta_N}{d_N} & \frac{\cos \theta_N}{d_N}
	\end{array}\right]_{N \times 2}.
\end{equation}
Then, the Fisher information matrix (FIM) of estimating the parameter vector ${{\bm{\psi}}}_t$ for the AOA-based MIMO radar considered is equal to
\begin{equation}\label{FIMexpression_1}
	\begin{aligned}
		{\bf{F}}_{\rm{A}} \! & = \! {\tilde {\mathbf{J}}}_{\mathrm{A}}^T \mathbf{\Sigma}_{\mathrm{A}}^{-1} {\tilde {\mathbf{J}}}_{\mathrm{A}} \\
		& =\! |\zeta_a |^2 \!\sum\nolimits_{j = 1}^N \! {\sum\nolimits_{i = 1}^N \! {\frac{{{{\cos }^2}{\theta _i}}}{{d_j^2 d_i^2}} \! \left[\! {\begin{array}{*{20}{c}}
						{\!\frac{{{{\sin }^2}{\theta _i}}}{{d_i^2}}}&\!{ - \frac{{\sin {\theta _i}\cos {\theta _i}}}{{d_i^2}}}\\
						\!{ - \frac{{\sin {\theta _i}\cos {\theta _i}}}{{d_i^2}}}&\!{\frac{{{{\cos }^2}{\theta _i}}}{{d_i^2}}}
				\end{array}} \!\right]} } \!,
	\end{aligned}
\end{equation}
where ${\tilde {\mathbf{J}}}_{\mathrm{A}} = \left[ \mathbf{J}_{\mathrm{A}}^T, \cdots, \mathbf{J}^T_{\mathrm{A}} \right]^T \in {\mathbb {C}}^{N^2 \times 2}$, $\mathbf{\Sigma}_{\mathrm{A}} = {\rm{diag}}(\rho_{1,1}^2,\cdots,\rho_{i,j}^2,\cdots,\rho_{N,N}^2)  \in {\mathbb {C}}^{N^2 \times N^2}$, and $|\zeta_a |^2 = \frac{1}{6}\pi^2   M_r\left(M_r^2-1 \right) G_t \sigma p^s \gamma_0 / \sigma_s^2$ \cite{Liu2022SurveyFundamental}.
Given the random location of ISAC BSs, the expected CRLB for any unbiased estimator of the target position is given by
\vspace{0mm}
\begin{equation}\label{CRLBAgnleExpression}	
	\mathrm{CRLB}_{\rm{A}}= {\rm{E}}_{\Phi_b, G_t} \left[\operatorname{tr}\left( {\mathbf{F}}_{\mathrm{A}}^{-1}\right)\right].
	\vspace{0mm}
\end{equation}
In (\ref{CRLBAgnleExpression}), the expectation operation accounts for the randomness in the locations of sensing BSs and the variability in beam power caused by user channel fluctuations, thereby representing the average sensing performance bound across the entire network.

\subsubsection{TOF Measurement based Localization}

From transmitter $j$ to the target and then to receiver $i$, the term \( d_{i,j} \) represents the corresponding distance between the \( j \)th transmitter and the \( i \)th receiver, which is given by
\begin{equation}
	\begin{aligned}
		&\hat d_{i,j} (\tau_{i,j}) = \\
		& \sqrt{\!\left(x_i - x_t\right)^2 + \left(y_i - y_t\right)^2} +  \sqrt{\left(x_j - x_t\right)^2 + \left(y_j - y_t\right)^2} + n^r_{i,j},
	\end{aligned}
\end{equation}
where $\quad n^r_{i,j} \sim \mathcal{N}\left(0, \eta_{i,j}^2\right)$ and $\eta_{i,j}^2 = \frac{3 c^2 \sigma_s^2  }{2 \pi^2 G_t M_r B^2  \gamma_{i,j}}$. $c$ denotes the speed of light, $B^2$ represents the squared effective bandwidth, showing that the larger bandwidth offers the more accurate TOF estimation. For the TOF-based range estimation \(\hat{d}_{i,j}\), we apply matched filtering to the received signal to correlate it with a replica of the transmitted waveform. This process highlights peaks corresponding to time delays caused by targets, which are then converted into range estimates using the speed of light.
For signals emanating from the transmitter $i$, the Jacobian matrices of the $N$ receiver measurement errors evaluated at the true target position ${{\bm{\psi}}}_t$ are given by  
\begin{equation}
	\begin{aligned}
		{\bf{J}}_{\rm{R}}^T = \left[ {\begin{array}{*{20}{c}}
				{\frac{{\partial {d_{11}}}}{{\partial {x_t}}}}& \cdots &{\frac{{\partial {d_{i,j}}}}{{\partial {x_t}}}}& \cdots &{\frac{{\partial {d_{NN}}}}{{\partial {x_t}}}}\\
				{\frac{{\partial {d_{11}}}}{{\partial {y_t}}}}& \cdots &{\frac{{\partial {d_{i,j}}}}{{\partial {y_t}}}}& \cdots &{\frac{{\partial {d_{NN}}}}{{\partial {y_t}}}}
		\end{array}} \right],
	\end{aligned}
\end{equation}
where $\frac{\partial {d}_{i,j}}{\partial x_t} = \cos {\theta _i} + \cos {\theta _j}$ and $\frac{\partial {d}_{i,j}}{\partial y_t} = \sin {\theta _i} + \sin {\theta _j}$. Let ${a_{i,j}} = \cos {\theta _i} + \cos {\theta _j}$ and ${b_{i,j}} = \sin {\theta _i} + \sin {\theta _j}$.
Then, the FIM of estimating the parameter vector ${{\bm{\psi}}}_t$ for the TOF measurement radar considered is equal to \cite{sadeghi2021target}
\begin{equation}\label{FIMexpression}
	\begin{aligned}
		{\bf{F}}_{\mathrm{R}} =& \mathbf{J}_{\mathrm{R}}^T  \mathbf{\Sigma}_{\mathrm{R}}^{-1}  \mathbf{J}_{\mathrm{R}}   \\
		=& |\zeta_r |^2 \sum\nolimits_{i = 1}^N \sum\nolimits_{j = 1}^N {{d}}_i^{ - \beta }{{{d}}}_j^{ - \beta } {  {\left[ {\begin{array}{*{20}{c}}{a_{i,j}^2}&{{a_{i,j}}{b_{i,j}}}\\
						{{a_{i,j}}{b_{i,j}}}&{b_{i,j}^2}
				\end{array}} \right]} } ,
	\end{aligned}
\end{equation}
where $\mathbf{\Sigma}_{\mathrm{R}} = {\rm{diag}}(\eta_{1,1}^2,\cdots,\eta_{i,j}^2,\cdots,\eta_{N,N}^2)$.
In (\ref{FIMexpression}), we have \cite{sadeghi2021target} $|\zeta_r |^2 = \frac{8 \pi^2 p^s G_t M_r B^2 \sigma \gamma_0}{ 3c^2 \sigma_s^2}$, where $|\zeta_r |$ is the common system gain term.
Given the random location of ISAC BSs, the expected CRLB for any unbiased estimator of the target position is given by
\vspace{0mm}
\begin{equation}	
	\mathrm{CRLB}_{\rm{R}} = {\rm{E}}_{\Phi_b, G_t} \left[\operatorname{tr}\left( {\mathbf{F}}_{\rm{R}}^{-1}\right)\right].
	\vspace{0mm}
\end{equation}

\subsubsection{Joint AOA and TOF Localization}
Incorporating both AOA and TOF measurements, rather than relying solely on one type of AOA or TOF measurement, can significantly enhance the accuracy and reliability of the localization system, namely the hybrid localization method. It is assumed that the AOA estimation errors and the TOF estimation errors are uncorrelated.\footnote{If there exists a correlation between these estimation errors, e.g., wideband and large array aperture for near-field target sensing, the joint FIM should include off-diagonal terms capturing their interdependence, which would alter the CRLB expression. In such cases, equation (\ref{LowerCRB}) serves as a lower-bound approximation.} Using both AOA and TOF measurements, the expected CRLB for any unbiased estimator of the target position is given by
\begin{equation}\label{LowerCRB}
	\mathrm{CRLB}_{\rm{H}} = {\rm{E}}_{\Phi_b, G_t} \left[\operatorname{tr}\left( \left({\bf{F}}_{\mathrm{A}} +  {\bf{F}}_{\mathrm{R}}\right)^{-1} \right)\right].
\end{equation}

\subsection{Cooperative Communication Model}

We assume that the transmitters use non-coherent joint transmission, where the useful signals are combined by accumulating their powers. In this work, we employ a user-centric clustering approach, where the BS closest to the typical user sends collaboration requests to other BSs within a range \(D\) of the user. The signal received at the typical user is then given by
\vspace{0mm}
	\begin{align}\label{TransmitSignal}
		y_{c}=& \underbrace{\sum\nolimits_{i \in \Phi_c }d_i^{-\frac{\alpha}{2}} \mathbf{h}_{i}^H \mathbf{W}_i \mathbf{s}_i}_{\text{collaborative intended signal}} +\underbrace{\sum\nolimits_{{j \in \{\Phi_b \backslash \Phi_c\}}}d_j^{-\frac{\alpha}{2}} \mathbf{h}_{j}^H \mathbf{W}_j \mathbf{s}_j}_{\text{inter-cluster interference}}  \nonumber \\
		& + {n_{c}},
		\vspace{0mm}
	\end{align}
where $\alpha \ge 2$ is the pathloss exponent, $\mathbf{h}^H_{i} \sim \mathcal{C N}\left(0, \mathbf{I}_{M_{\mathrm{t}}}\right)$ is the channel vector of the link between the BS at $\mathbf{d}_i$ to the typical user, $\Phi_c$ is the cooperative BS set, and $n_{c}$ denotes the noise. In \eqref{TransmitSignal}, $\mathbf{s}_i$ denotes the transmitted signal vector from the cooperating BSs within the cluster (with the sensing signals nulled at the user due to ZF beamforming), while the second term accounts for interference from BSs outside the cooperation cluster. We focus on evaluating the performance of an interference-limited network within dense cellular scenarios. The impact of noise is disregarded in this analysis as the interference arriving from outside the cooperation region far exceeds the negligible noise power in our scenarios. The evaluation is based on the signal-to-interference ratio (SIR) \cite{Park2016OptimalFeedback}.
The SIR of the received signal at the typical user can be expressed as
\vspace{0mm}
\begin{equation}\label{SIRexpression}
	{\rm{SIR}}_c = \frac{ {\sum\limits_{i \in \Phi_c} {{g_i}{{d}_i^{ - \alpha }}} }}{{\sum\limits_{j \in \{\Phi_b \backslash \Phi_c\}}  {{g_j}} {{d}_j^{ - \alpha }}}},
	\vspace{0mm}
\end{equation}
where ${\sum\limits_{i \in \Phi_c} {{g_i}{{d}_i^{ - \alpha }}} }$ is the cooperative desired signal, ${\sum\limits_{j \in \{\Phi_b \backslash \Phi_c\}}  {{g_j}} {{d}_j^{ - \alpha }}}$ represents the inter-cluster interference. Finally, the interference channel's gain is $g_{j} = p^c\left|\mathbf{h}_{j}^H \mathbf{w}_j^c\right|^2 + p^s\left|\mathbf{h}_{j}^H \mathbf{w}_j^s\right|^2$, accounting for interference from both sensing and communication signals transmitted outside the cooperation region.
The average data rate of users is given by 
\vspace{0mm}
\begin{equation}
	R_c=\mathrm{E}_{\Phi_b,g_i}[\log (1+\mathrm{SIR}_c)].
	\vspace{0mm}
\end{equation}

\section{Sensing Performance Analysis}
\label{SensingSection}

In the following, we first analyze both AOA-based and TOF-based localization. Then, we further examine the performance of the hybrid localization method. 
To facilitate the analysis, we assume that the number of BSs within a range \(D\) from each target equals the average number of BSs in the area obeying the density of the PPP. This approximation is widely used and matches the law of large numbers in dense deployments. Specifically, the circle of radius \(D\) centered at a target contains exactly \(|\Phi_c| = \lambda_b \pi D^2\) of BSs. In simulations, we verify that this assumption greatly simplifies the derivation, while providing tight approximations.

\subsection{Angle Measurement Based Localization}

\subsubsection{Geometry Gain of Cooperative Sensing}

In practice, it is non-trivial to derive a tractable expression for the expected CRLB due to the complex operations involved, including matrix inversions. 
To this end, we first study the sensing gain brought by geometric diversity,\footnote{Geometric diversity refers to the improvement in localization accuracy achieved by positioning multiple BSs at different angles around a target, allowing for more reliable and complementary directional measurements.} where all ISAC BSs are assumed to be placed at the same distance from the target, i.e. $d_i = d_j$, $\forall i,j, \in {\cal{N}}$, and ${\cal{N}} = \{1,\cdots,N\}$. Then, by ignoring the effect of each link's signal strength, i.e., $|\zeta_a |$ and $d_i$, we analyze the cooperative localization performance improvement purely gleaned from the direction diversity of BSs. To this end, we define a new matrix based on the FIM by removing $|\zeta_a |$ and $d_i$ in (\ref{FIMexpression}), yielding
\vspace{-2mm}
\begin{equation}\label{FIMEquation}
	\tilde {\bf{F}}_{\mathrm{A}} \! =  \! \sum\nolimits_{j = 1}^N  \! {\sum\nolimits_{i = 1}^N  \! {{{\cos }^2}{\theta _i} \!\left[ {\begin{array}{*{20}{c}}
					 \!{{{\sin }^2}{\theta _i}}& \!{ - \sin {\theta _i}\cos {\theta _i}}\\
					 \!{ - \sin {\theta _i}\cos {\theta _i}}& \!{{{\cos }^2}{\theta _i}}
			\end{array}}  \!\right]} }  \!.
\end{equation}
Following the definition in \cite{guvenc2009survey}, let ${\rm{tr}}({ \tilde {\bf{F}}^{-1}_{\mathrm{A}}})$ be termed the geometric dilution of precision (GDoP) in AOA-based radar systems, which can be formulated as follows:
\vspace{0mm}
\begin{equation}\label{GDoPExpression}
	\begin{aligned}
		&{\rm{tr}}({\tilde {\bf{F}}}_{\mathrm{A}}^{-1}) = \\
		&\frac{\sum\nolimits_{i = 1}^N {\sum\nolimits_{j = 1}^N {e_{i}^2} }  + \sum\nolimits_{i = 1}^N {\sum\nolimits_{j = 1}^N {f_{i}^2} } }{\left(\! {\sum\nolimits_{i = 1}^N {\sum\nolimits_{j = 1}^N {e_{i}^2} } } \!\right)\!\left(\! {\sum\nolimits_{i = 1}^N {\sum\nolimits_{j = 1}^N {f_{i}^2} } } \right) \! - \! {\left( {\sum\nolimits_{i = 1}^N {\sum\nolimits_{j = 1}^N {{e_{i}}{f_{i}}} } } \!\right)^2}}.
		\vspace{0mm}
	\end{aligned}
\end{equation}
where $e_{i} = {{\sin }}{\theta _i} {{\cos }}{\theta _i}$ and $f_{i} = {{\cos }^2}{\theta _i}$. Then, we employ tight approximations to derive a simplified expression, aiming to provide an intuitive illustration of the benefits of  geometric diversity. %Characterizing the geometric gain from cooperative localization directly is challenging due to the unpredictable sensing directions. 

\begin{Pro}\label{GDoPDerivation}
	The expected GDoP can be approximated as
	\vspace{0mm}
	\begin{equation}\label{GDoPExpression1}
		\begin{aligned}
			{\rm{E}}_{\theta}\left[ {{\rm{tr}}\left( {\tilde {\bf{F}}}_{\mathrm{A}}^{-1} \right)} \right] \approx  \frac{{32}}{{3N\left( {N - 1} \right)}}.
			\vspace{0mm}
		\end{aligned}
	\end{equation}
\end{Pro}
\begin{proof}
	Please refer to Appendix A.
\end{proof}

Building upon the conclusion in Proposition \ref{GDoPDerivation}, the scaling law associated with an infinite number of ISAC BSs can be derived as follows.
\begin{Cor}{\label{ScalingLawWithoutDis}}
	For an infinite number of BSs involved in cooperative sensing, the expected GDoP can be further reformulated from (\ref{GDoPExpression1}) as
	\vspace{0mm}
	\begin{equation}\label{GeometryGain}
		\mathop {\lim }\limits_{N \to \infty }  \!{\rm{E}}_{\theta}\left[ {{\rm{tr}} \!\left( \! {\tilde {\bf{F}}}_{\mathrm{A}}^{-1} \right)} \right] N^2  \!\approx  \! \mathop {\lim }\limits_{N \to \infty } \frac{{32}}{{3N\left( {N - 1} \right)}} N^2  \!= \! \frac{32}{{{3}}}.
		\vspace{0mm}
	\end{equation}
\end{Cor}
Corollary \ref{ScalingLawWithoutDis} states that the expected GDoP is inversely proportional to the square of the number of BSs involved. %The analysis above provides insights into the geometric gain for cooperative sensing by only considering random sensing directions. In the following section, we will further detail the expected CRLB derivation and compare the localization gain and geometric gain.
If each antenna array can orient itself perpendicular to the observation direction, then the overall performance of a cooperative sensing system is likely to improve. To this end, we define a new matrix based on the FIM by removing $\cos^2 \theta_i = 1$ in (\ref{FIMEquation}), yielding
\vspace{-2mm}
\begin{equation}
	\tilde {\bf{F}}_{\mathrm{AO}} = \sum\nolimits_{j = 1}^N {\sum\nolimits_{i = 1}^N {\left[ {\begin{array}{*{20}{c}}
					{{{\sin }^2}{\theta _i}}&{ - \sin {\theta _i}\cos {\theta _i}}\\
					{ - \sin {\theta _i}\cos {\theta _i}}&{{{\cos }^2}{\theta _i}}
			\end{array}} \right]} } .
\end{equation}
\begin{Pro}\label{GDoPDerivation2}
	The GDoP gain attained from dynamically controlling the orientation of the antenna array is $\frac{8}{3}$. 
\end{Pro}
\begin{proof}
	Similar to the proof details adopted in Appendix A, the expected GDoP of $\tilde {\bf{F}}_{\mathrm{AO}}$ can be approximated as
	\vspace{0mm}
	\begin{equation}\label{GDoPExpression2}
		\begin{aligned}
			{\rm{E}}_{\theta}\left[ {{\rm{tr}}\left( \tilde {\bf{F}}_{\mathrm{AO}}^{-1} \right)} \right] \approx  \frac{{4}}{{N^2 { - N} }}.
			\vspace{0mm}
		\end{aligned}
	\end{equation}
	It can be readily verified that the GDoP gain attained from orientation control is $\frac{8}{3}$ as compared to Proposition \ref{GDoPDerivation}.
\end{proof}

Proposition \ref{GDoPDerivation2} demonstrates that by appropriately controlling the orientation of each antenna array, a consistent gain in the scaling law can be achieved.

\subsubsection{Performance Gain of Cooperative Sensing}
In this subsection, we derive the closed-form CRLB expression under the assumption of random locations of both the BSs and targets. The CRLB expression can be equivalently transformed into (\ref{CRLBexpression}), as shown at the top of the page.% Therefore, the final CRLB contains three parts: Transmit beamforming gain, receive beamforming gain, and the geometry gain. 
\begin{figure*}
	\begin{equation}\label{CRLBexpression}
		\begin{aligned}
			{\mathrm{CRLB}}_{\rm{A}} = {{\rm{E}}_{\Phi_b,G_t}}\left[ \frac{|\zeta_a |^{-2} {\sum\nolimits_{j = 1}^N {{{d_j^{-\beta}}}\sum\nolimits_{i = 1}^N {{{d_i^{-\beta-2}}}\left( e_i^2 + f_i^2 \right)} } }}{{\left( {\sum\nolimits_{j = 1}^N  \!{\sum\nolimits_{i = 1}^N {{{e_i^2}}{{d_j^{-\beta}d_i^{-\beta-2}}}} } } \right)  \! \left( {\sum\nolimits_{j = 1}^N {\sum\nolimits_{i = 1}^N {{{f_i^2}}{{d_j^{-\beta}d_i^{-\beta-2}}}} } } \right) \! - \! {{\left( {\sum\nolimits_{j = 1}^N  \! {\sum\nolimits_{i = 1}^N {{{{e_i}{f_i}}}{{d_j^{-\beta}d_i^{-\beta-2}}}} } } \right) \!}^2}}}\right].
		\end{aligned}
	\end{equation}
%	\hrulefill
\end{figure*}
To obtain a more tractable CRLB expression, we resort to a simple yet tight approximation. Then the following conclusion is proved.
\begin{Pro}\label{SimplifiedWithDis1}
	For an infinite cooperative cluster size $N$ and fixed $|\zeta_a |$, the expected CRLB can be approximated as 
	\vspace{0mm}
	\begin{equation}\label{SimplifiedExpressionCRLB}
		{\rm{CRLB}}_{\rm{A}} \approx	\frac{16|\zeta_a |^{-2}{\sum\nolimits_{i = 1}^N {{{{{\rm{E}}[d_i]}^{-\beta-2}}}} }}{{3 {\sum\nolimits_{k = 1}^N {{{{{\rm{E}}[d_k]}^{-\beta}}}} }  {\sum\nolimits_{i = 1}^N {\sum\nolimits_{i > j}^N {{{{{\rm{E}}[d_i]}^{-\beta-2}{{\rm{E}}[d_j]}^{-\beta-2}}}} } } }}.
		\vspace{0mm}
	\end{equation}
\end{Pro}
\begin{proof}
Please refer to Appendix B.
\end{proof}

Interestingly, we found that the expected CRLB in Proposition \ref{SimplifiedWithDis1} is only determined by the expected distance from the BS to the typical target. %It maybe readily verified that (\ref{SimplifiedExpressionCRLB}) achieves a good approximation by Monte Carlo simulations, as shown in Section \ref{SimulationsSection}.
Furthermore, the expected distance from the $n$th closest BS to the typical target can be expressed as
\vspace{0mm}
\begin{equation}\label{ExpectedDistance}
	{\rm{E}} \left[ {{d_n}} \right] =  { {\frac{{\Gamma \left( n + \frac{1}{2} \right)}}{\sqrt{\lambda_b \pi } \Gamma (n)}} } \approx \sqrt{\frac{{n}}{\lambda_b \pi}}.
	\vspace{0mm}
\end{equation}
By substituting (\ref{ExpectedDistance}) into (\ref{SimplifiedExpressionCRLB}), the CRLB expression can be further approximated as
\vspace{0mm}
\begin{equation}\label{CRLB_expression}
	{\rm{CRLB}}_{\rm{A}} \! \approx \! \frac{32|\zeta_a |^{-2}{\sum\nolimits_{i = 1}^N {{i^{ - \frac{\beta}{2} - 1}}} }}{{{3}{{}}{\lambda_b ^3}{\pi ^3} {\sum\nolimits_{k = 1}^N {{k^{ - \frac{\beta}{2} }}} } \left( {{{\left( {\sum\nolimits_{i = 1}^N {{i^{ -  \frac{\beta}{2} -1}}} } \right)}^2}  \!- \! \sum\nolimits_{i = 1}^N {{i^{ - {\beta}- 2}}} }  \!\right)}} \!.
	\vspace{0mm}
\end{equation}
For $\beta = 2$, we further derive the scaling law of the localization accuracy as follows.

\begin{theorem}\label{SimplifiedWithDis3}
	For an infinite cooperative cluster size $N$ and fixed $|\zeta_a |$, the expected CRLB of AOA-based localization is given by
	\vspace{0mm}
	\begin{equation}\label{EquationForInfityN}
		\mathop {\lim }\limits_{N \to \infty }{\rm{CRLB}}_{\rm{A}} \times {{\ln }}N  \approx \frac{{320}}{{3|\zeta_a |^{2}{\lambda_b ^3}{\pi ^5} }}.
		\vspace{0mm}
	\end{equation}
\end{theorem}
\begin{proof}
	Please refer to Appendix C.
\end{proof}

\begin{remark}
The CRLB scaling law for random BS locations, as presented in Theorem \ref{SimplifiedWithDis3}, is critical for cooperative sensing design. Unlike the scenario described in Proposition \ref{GDoPDerivation}, the performance gain decreases as more BSs participate. This gain erosion occurs because, although distant BSs do contribute to sensing diversity, their measurements have a limited impact on localization accuracy due to the inevitable propagation loss.
\end{remark}

For optimal sensing performance, the transmit beamforming gain can be approximated as ${\left\lfloor\frac{{{\lambda _t} D^2 \pi}}{N}\right\rfloor}$, and the BS density can be denoted by $\frac{N}{\pi D^2}$. Then, we have
\begin{equation}\label{BSPowerConstraints}
	\begin{aligned}
		{\rm{CRLB}}_{\rm{A}} \approx & \frac{1}{|\tilde \zeta_a |^2{\underbrace {{{\left\lfloor {\frac{{{\lambda _r}D^2 \pi}}{N}} \right\rfloor}^3}}_{{\text{Receive gain}}} \times \underbrace {\left\lfloor\frac{{{\lambda _t}D^2 \pi}}{N}\right\rfloor}_{{\text{Transmit gain}}} \times \underbrace{\frac{{{N^3}}}{{{({D^2 \pi})^3}}}{\pi ^5} \times \ln N}_{\text{Geometry gain}}}} \\
		\approx & \frac{320N}{3|\tilde \zeta_a |^2 D^2 \pi^6 {\lambda _t} {\lambda _r} \ln N},
	\end{aligned}
\end{equation}
where $|\tilde \zeta_a |^2 = \frac{1}{6}\pi^2   \sigma p^s \gamma_0 / \sigma_s^2$. According to (\ref{BSPowerConstraints}), as the number of BSs increases, the value of ${\rm{CRLB}}_{\rm{A}}$ also increases monotonically with $N$. Therefore, under total transmit power constraints, a fully distributed antenna allocation is unlikely to be optimal. This is primarily because accurate AOA measurement relies on multiple antennas to enhance estimation accuracy.

\subsection{Range Measurement Based Localization}

In this subsection, we derive the closed-form CRLB expression of the TOF-based localization method. Similarly, we define the corresponding GDoP expression as follows:
%\vspace{-2mm}
\begin{equation}
	\tilde {\bf{F}}_{\mathrm{R}} = \sum\nolimits_{i = 1}^N {\sum\nolimits_{j = 1}^N {\left[ {\begin{array}{*{20}{c}}
					{a_{i,j}^2}&{{a_{i,j}}{b_{i,j}}}\\
					{{a_{i,j}}{b_{i,j}}}&{b_{i,j}^2}
			\end{array}} \right]} } .
	%		\vspace{0mm}
\end{equation}
\begin{Pro}\label{GDoPDerivationRangeBased}
	The expected GDoP for TOF-based localization can be approximated as
	\vspace{0mm}
	\begin{equation}\label{GDoPExpression1Range}
		\begin{aligned}
			{\rm{E}}_{\theta}\left[ {{\rm{tr}}\left( {\tilde {\bf{F}}}_{\mathrm{R}}^{-1} \right)} \right] \approx  \frac{{2}}{{N\left( {N - 1} \right)}}.
			\vspace{0mm}
		\end{aligned}
	\end{equation}
\end{Pro}
\begin{proof}
    It can be proved in a similar way as that in Proposition \ref{GDoPDerivation}. The details are omitted due to page limitation.
\end{proof}

To facilitate the performance analysis, the CRLB expression can be equivalently transformed into (\ref{CRLBExpression_simplify}),
\begin{figure*}
	\begin{equation}\label{CRLBExpression_simplify}
		{\rm{CRLB}}_{\rm{R}} = {{\rm{E}}_{\Phi_b}}\bigg[|\zeta_r |^{-2} \times 
		\frac{{2\sum\nolimits_{i = 1}^N \!{\sum\nolimits_{j = 1}^N d_i^{-\beta}d_j^{ - \beta}\left(1+\cos \left( {{\theta _i} - {\theta _j}} \right)\right) } } } {{\!\sum\nolimits_{l = 1}^N {\! \sum\nolimits_{k = 1}^N {\!\! \sum\nolimits_{i \ge k}^N {\! \sum\nolimits_{j > \!{\lceil (k - i)N + l \rceil \!}^+}^N \!{(d_i d_j d_l d_k)^{-\beta}} } } } {{\! \left( {{a_{k,l}}{b_{i,j}} \!-\! {a_{i,j}}{b_{k,l}}} \right)}^2}}}\! \bigg] \!,
	\end{equation}
\end{figure*}
where ${\lceil x \rceil}^+ = \max(x,1)$.
For $\beta = 2$, we further derive the scaling law of the localization accuracy as follows.

\begin{theorem}\label{SimplifiedWithDisRange}
	For an infinite cooperative cluster size $N$ and fixed $|\zeta_r |$, the expected CRLB is given by
	\vspace{0mm}
	\begin{equation}
		\mathop {\lim }\limits_{N \to \infty }{\rm{CRLB}}_{\rm{R}} \times {{\ln }^2}N \approx \frac{2}{|\zeta_r |^{2}{{\lambda_b ^2}{\pi ^2}}}.
		\vspace{0mm}
	\end{equation}
\end{theorem}
\begin{proof}
	The proof follows a similar approach to that in Theorem \ref{SimplifiedWithDis3}. Details are omitted due to space constraints.  
\end{proof}

\begin{remark}
	In comparison to the scaling law for TOF-based localization, AOA-based localization exhibits a less favorable scaling law, as demonstrated in Theorems \ref{SimplifiedWithDis3} and \ref{SimplifiedWithDisRange}. This is due to the fact that the Jacobian matrix in AOA-based localization includes a higher-order distance attenuation coefficient. Specifically, when converting AOA measurements to target location, the localization error is inversely proportional to the distance between the BS and the target, leading to a more significant reduction in accuracy as both the number of transceivers and the measurement distance increase. Consequently, with greater measurement distances, the accuracy of AOA-based localization diminishes, resulting in a less favorable scaling law compared to TOF-based localization.
\end{remark}

The transmit beamforming gain can be approximated as ${\left\lfloor\frac{{{\lambda _t} D^2 \pi}}{N}\right\rfloor}$. For total transmit power constraints, we have
\begin{equation}\label{RangeBSPowerConstraints}
	\begin{aligned}
		{\rm{CRLB}}_{\rm{R}} 
	\approx & \frac{1}{|\tilde \zeta_r |^2{\underbrace {{{\left\lfloor {\frac{{{\lambda _r}D^2 \pi}}{N}} \right\rfloor}}}_{{\text{receive gain}}} \times \underbrace {\left\lfloor\frac{{{\lambda _t}D^2 \pi}}{N}\right\rfloor}_{{\text{transmit gain}}} \times \underbrace{\frac{{{N^2}}}{{{{D^2 \pi}^2}}}{\pi ^2} \times \ln^2 N}_{\text{Geometry gain}}}} \\
	\approx  & \frac{2}{ |\tilde \zeta_r |^{2} D^2 \pi^2 {\lambda _t} {\lambda _r}  \ln^2 N },
	\end{aligned}
\end{equation}
where $|\tilde \zeta_r |^2 = \frac{8 \pi^2 p^s B^2 \sigma \gamma_0}{ c^2 \sigma_s^2}, $
According to (\ref{RangeBSPowerConstraints}), when the number of BSs is sufficiently large, the ${\rm{CRLB}}_{\rm{R}}$ value decreases monotonically as the number of BSs \(N\) increases. Therefore, given the total BS power constraint, the TOF-based localization method tends to favor a distributed antenna allocation to achieve better sensing results at closer distances.

\subsection{Joint Angle and Range Localization}

In localization systems, fusing information from different measurement modalities can significantly improve estimation accuracy. In particular, combining AOA and TOF estimates consistently enhances performance by exploiting their complementary information. The following theorem formally establishes that the CRLB of the hybrid localization, which results from this combination, is always lower than or equal to the CRLB from either individual measurement.

\begin{thm}\label{IncreasingInformation}
	$\mathrm{CRLB}_{\rm{H}} \le \min \left(\mathrm{CRLB}_{\rm{A}},\mathrm{CRLB}_{\rm{R}}\right)$ .
\end{thm}
\begin{proof}
	First, we will prove the inequality  when \( \mathbf{F}_{\rm{R}} \) is a rank-one matrix, i.e.,
	\begin{equation}
	\operatorname{Tr}\left((\mathbf{F}_{\rm{A}}+  \lambda \mathbf{u} \mathbf{u}^{T}  )^{-1}\right) \leq \operatorname{Tr}\left(\mathbf{F}_{\rm{A}}^{-1}\right).
	\end{equation}  
	where \( \mathbf{u} \in \mathbb{R}^{n \times 1} \) is a nonzero vector, and \( \lambda > 0 \) is a scalar.
	Then, the inverse of the sum \( \mathbf{F}_{\rm{A}} + \mathbf{F}_{\rm{R}} \) follows the Sherman-Morrison formula:
	\begin{equation}
	(\mathbf{F}_{\rm{A}}+\mathbf{F}_{\rm{R}})^{-1} = \mathbf{F}_{\rm{A}}^{-1} - \frac{\lambda \mathbf{F}_{\rm{A}}^{-1} \mathbf{u} \mathbf{u}^{T} \mathbf{F}_{\rm{A}}^{-1}}{1 + \lambda \mathbf{u}^{T} \mathbf{F}_{\rm{A}}^{-1} \mathbf{u}}.
	\end{equation}
	Since  
	$
	\operatorname{Tr} \left( \mathbf{F}_{\rm{A}}^{-1} \mathbf{u} \mathbf{u}^{T} \mathbf{F}_{\rm{A}}^{-1} \right) = \|\mathbf{F}_{\rm{A}}^{-1} \mathbf{u}\|^2 \geq 0,
	$
	it follows that:
	\begin{equation}
	\operatorname{Tr} \left( (\mathbf{F}_{\rm{A}}+\mathbf{F}_{\rm{R}})^{-1} \right) \leq \operatorname{Tr} (\mathbf{F}_{\rm{A}}^{-1}).
\end{equation}
	Thus,  we have
	\begin{equation}
	\mathrm{CRLB}_{\rm{H}} \leq \mathrm{CRLB}_{\rm{A}}.
	\end{equation}
	For a general $\mathbf{F}_{\rm{R}}$ that is not rank-one, it can be decomposed as a sum of multiple rank-one matrices. By applying the same reasoning iteratively, the result remains valid. Similarly, we can prove $\mathrm{CRLB}_{\rm{H}} \leq \mathrm{CRLB}_{\rm{R}}$.
\end{proof}

Lemma \ref{IncreasingInformation} demonstrates that combining AOA and TOF measurements improves the target localization performance by enhancing the information available.

In a manner analogous to AOA-only and TOF-only localization methods, we analyze the GDoP for the hybrid localization method in the following. We define the GDoP for this hybrid approach as
\begin{equation}
	{{\tilde {\bf{F}}}_{{\rm{H}}}} \!=\! \sum\nolimits_{j = 1}^N \!{\sum\nolimits_{i = 1}^N \!{\left[ {\begin{array}{*{20}{c}}
					{a_{i,j}^2 + c_{i,j}^2}&{{a_{i,j}}{b_{i,j}} - {c_{i,j}}{e_{i,j}}}\\
					{{a_{i,j}}{b_{i,j}} - {c_{i,j}}{e_{i,j}}}&{b_{i,j}^2 + e_{i,j}^2}
			\end{array}}\! \right]} } .
\end{equation}
\begin{Pro}\label{GDoPDerivationHybrid}
	The expected GDoP can be approximated as
	\vspace{0mm}
	\begin{equation}\label{GDoPExpressionHybrid}
		\begin{aligned}
			{\rm{E}}_{\theta}\left[ {{\rm{tr}}\left( {\tilde {\bf{F}}}_{\mathrm{A}}^{-1} \right)} \right] \approx  \frac{{160}}{{99{N^2} - 67}}.
			\vspace{0mm}
		\end{aligned}
	\end{equation}
\end{Pro}
\begin{proof}
	It can be readily proved in a similar way to Appendix A.
\end{proof}

It is straightforward to see that the GDoP for the hybrid localization method that combines both TOF and AOA measurements is lower than that for methods relying solely on either AOA or TOF measurements. However, because both types of measurements are obtained from the same nodes, the improvement in geometric gain compared to the GDoP presented in Proposition \ref{GDoPDerivation} is relatively modest.

It is worth noting that localization methods solely based on AOA or TOF measurements require at least two nodes to obtain valid localization results. However, the hybrid localization method utilizing both AOA and TOF measurements can obtain location information even for a single ISAC BS. 
It is worth noting that when there is only one BS, it is impossible to obtain the target's location information by relying solely on TOF or AOA measurements, i.e., we have $\mathrm{CRLB}_{\rm{A}} = \mathrm{CRLB}_{\rm{R}} = \infty$. The localization method that combines TOF and AOA measurements can successfully locate the target by relying on only a single BS.
When $N=1$, we have
\begin{equation}
	{{\bf{F}}_{{\rm{H}}}} = \frac{1}{{d_i^2 d_j^2}}\left[ {\begin{array}{*{20}{c}}
			{a_{i,j}^2 + c_{i,j}^2}&{{a_{i,j}}{b_{i,j}} - {c_{i,j}}{e_{i,j}}}\\
			{{a_{i,j}}{b_{i,j}} - {c_{i,j}}{e_{i,j}}}&{b_{i,j}^2 + e_{i,j}^2}
	\end{array}} \right].
\end{equation}
Then, the CRLB can be formulated as:
\begin{equation}
	{\rm{tr}} \left({\bf{F}}_{{\rm{H}}}^{-1}\right) = \frac{1}{{\left| {{\zeta _R}} \right|{{\cos }^2}{\theta _i}}} + \frac{1}{4 \left|  {{\zeta _A}} \right| {d_i^2}}.
\end{equation}

It is evident that, compared to TOF-based localization, the accuracy of AOA-based localization is more sensitive to the distance of the BS from the target. Specifically, as the BS moves farther from the target, a larger antenna array is needed to compensate for the increased path loss. However, AOA-based localization methods have the advantage of requiring less stringent time synchronization between BSs, allowing them to achieve accurate localization through collaboration even without precise synchronization.
When the number of cooperating BSs increases greatly, the sensing performance will be determined by the TOF measurement method. 

\begin{figure*}
	\begin{equation}\label{CRLBexpressionHybrid}
		\small
		\begin{aligned}
			&{\rm{tr}} \left(\left({\bf{F}}_{\mathrm{A}} +  {\bf{F}}_{\mathrm{R}}\right)^{-1}\right)  \\
			&= \! \frac{{\sum\nolimits_{i = 1}^N \! {\sum\nolimits_{j = 1}^N {\! \left( {\tilde a_{i,j}^2 + \tilde c_{i,j}^2} \right)} }  + \sum\nolimits_{i = 1}^N \! {\sum\nolimits_{j = 1}^N \! {\left( {\tilde b_{i,j}^2 + \tilde e_{i,j}^2} \right)} } }}{{\sum\nolimits_{l = 1}^N {\sum\nolimits_{k = 1}^N {\sum\nolimits_{i \ge k}^N {\sum\nolimits_{j > {{\left\lceil {(k - i)N + l} \right\rceil }^ + }}^N {\left( {{{\left( {{{\tilde a}_{i,j}}{{\tilde b}_{k,l}} - {{\tilde a}_{k,l}}{{\tilde b}_{i,j}}} \right)}^2}{\rm{ + }}{{\left( {{{\tilde c}_{i,j}}{{\tilde e}_{k,l}} - {{\tilde e}_{k,l}}{{\tilde c}_{i,j}}} \right)}^2}} \right)} } } }  + \!  \sum\nolimits_{l = 1}^N \!  {\sum\nolimits_{k = 1}^N {\sum\nolimits_{i = 1}^N \!  {\sum\nolimits_{j = 1}^N {{{\left( {{{\tilde a}_{i,j}}{{\tilde e}_{k,l}} + {{\tilde b}_{k,l}}{{\tilde c}_{i,j}}} \right)}^2}} } } } }}.
		\end{aligned}
	\end{equation}
%	\hrulefill
\end{figure*}
Under general setup, the CRLB expression can be transformed into (\ref{CRLBexpressionHybrid}), as shown at the top of the next page, where ${\rho _{i,j}} = \frac{1}{{d_j^2 d_i^2}}$, $\tilde a_{i,j} = \sqrt {{\rho _{i,j}}\left| {{\zeta _R}} \right|} \left( {\cos {\theta _i} + \cos {\theta _j}} \right)$, $\tilde b_{i,j} = \sqrt {{\rho _{i,j}}\left| {{\zeta _R}} \right|} \left( {\sin {\theta _i} + \sin {\theta _j}} \right)$, $\tilde c_{i,j} = \sqrt {{\rho _{i,j}}\left| {{\zeta _R}} \right|} \frac{{\sin {\theta _i}\cos {\theta _i}}}{{ {d_i} }}$, and $ \tilde e_{i,j} = \sqrt {{\rho _{i,j}}\left| {{\zeta _R}} \right|} \frac{{{{\cos }^2}{\theta _i}}}{{d_i}}$. To facilitate the analysis, we adopt the following approximation:
\begin{equation}
	\begin{aligned}
		&{\rm{E}}\left[ {{{\left( {{\tilde a_{k,l}}{\tilde b_{i,j}} - {\tilde a_{i,j}}{\tilde b_{k,l}}} \right)}^2}} \right] %\\
		%&= {\rm{E}}\left[ {{\rho _{i,j}}{\rho _{k,l}}{{\left| {{\zeta _R}} \right|}^2}({{ {\sin  \Delta_{k,i}  + \sin  \Delta_{k,j}  + \sin  \Delta_{l,i}  + \sin  \Delta_{l,j} } }})^2} \right] \\
		\approx 2{\rho _{i,j}}{\rho _{k,l}}{\left| {{\zeta _R}} \right|^2}.
	\end{aligned}
\end{equation}
Similarly, we have ${\rm{E}}\left[ {{{\left( {{\tilde c_{i,j}}{\tilde e_{k,l}} - {\tilde e_{k,l}}{\tilde c_{i,j}}} \right)}^2}} \right] = \frac{3}{{{\rm{32}}}}{\rho _{i,j}}{\rho _{k,l}}{\left| {{\zeta _R}} \right|^2}\frac{1}{{d_i^2d_k^2}}$, and ${\rm{E}}\left[ {{{\left( {{\tilde a_{i,j}}{\tilde e_{k,l}} + {\tilde b_{k,l}}{\tilde c_{i,j}}} \right)}^2}} \right] = {{\rho _{i,j}}{\rho _{k,l}}\left| {{\zeta _R}} \right|\left| {{\zeta _A}} \right|\frac{1}{{2d_k^2}}}$.

\begin{Pro}\label{ScalingLawHybrid}
	The CRLB of our hybrid localization method can be expressed as:
	\begin{equation}
		{\rm{CRLB}}_{\rm{H}}  \approx \frac{{24}}{{12\left| {{\xi _R}} \right|\lambda _b^2{\pi ^2}{{\ln }^2}N + \lambda _b^3{\pi ^5}\left| {{\xi _A}} \right|\ln N}}.
	\end{equation}
\end{Pro}
\begin{proof}
	Similar to the proof in Appendix C, the CRLB of the hybrid localization method can be expressed as in (\ref{HybridExpression}). Then, by substituting $\sum\nolimits_{i = 1}^N {{i^{ - 1}}}  \approx \ln N + \gamma  + \frac{1}{{2N}}$ and $\sum\nolimits_{i = 1}^N {{i^{ - 2}}}  \approx \frac{{{\pi ^2}}}{6}$ into (\ref{HybridExpression}), along with $\gamma = 0.577$, when $N\to \infty$, we have ${\rm{CRLB}}_{\rm{H}}  \approx \frac{{24}}{{12\left| {{\xi _R}} \right|\lambda _b^2{\pi ^2}{{\ln }^2}N + \lambda _b^3{\pi ^5}\left| {{\xi _A}} \right|\ln N}}$. This thus completes the proof.
	\begin{figure*}
		\begin{equation}\label{HybridExpression}
			\begin{aligned}
				{\rm{tr}}\! \left(\!\left({\!\bf{F}}_{\mathrm{A}} +  {\bf{F}}_{\mathrm{R}}\right)^{-1}\right) &\!\approx\! \frac{{2\left| {{\zeta _R}} \right|\left( {\ln N + \gamma  + \frac{1}{{2N}}} \right) +  \frac{{{\pi ^3{\lambda _b}}}}{{12}} \left| {{\zeta _A}} \right|}}{{{{\left| {{\zeta _R}} \right|}^2}{{ {\pi^2 {\lambda^2 _b}} }}{{\left( {\ln N + \gamma  + \frac{1}{{2n}}} \right)}^3} + {{}}{\frac{{{ {\pi^8 {\lambda^4 _b}} }}}{{1280}}}{{\left| {{\zeta _A}} \right|}^2}\left( {\ln N + \gamma  + \frac{1}{{2N}}} \right) + {{ {} }}\frac{{{\pi^5 {\lambda^3 _b}}}}{12}\left| {{\zeta _R}} \right|\left| {{\zeta _A}} \right|{{\left( {\ln N + \gamma  + \frac{1}{{2N}}} \right)}^2}}}. %\\
				%&\approx \frac{{2\left| {{\zeta _R}} \right|}}{{{{\left( {\pi {\lambda _b}} \right)}^2}{{\left| {{\zeta _R}} \right|}^2}{{\left( {\ln N + \gamma  + \frac{1}{{2N}}} \right)}^2} + {{\left( {\pi {\lambda _b}} \right)}^4}{\frac{{{\pi ^4}}}{{1280}}}{{\left| {{\zeta _A}} \right|}^2} + {{\left( {\pi {\lambda _b}} \right)}^3}\frac{{{\pi ^2}}}{{12}}\left| {{\zeta _R}} \right|\left| {{\zeta _A}} \right|\left( {\ln N + \gamma  + \frac{1}{{2N}}} \right)}} .
			\end{aligned}
		\end{equation}
	\end{figure*}
\end{proof}

\section{Communication Performance}
\label{CommunicationSection}
In this section, we derive a tractable expression of communication performance, and present an approximate expression to acquire the optimal antenna-to-BS allocation for cooperative transmission.

\subsection{Expression of Communication Rate}
\label{CommunicationPerformance}

According to \cite{hamdi2010useful}, for the uncorrelated variables $X$ and $Y$, it follows that:
\vspace{0mm}
\begin{equation}\label{CommunicationBasicEquation}
	{\rm{E}}\left[ {\log \left( {1 + \frac{X}{Y}} \right)} \right] \! = \! \int_0^\infty  {\frac{1}{z}} \left( {1 - {\rm{E}}\left[{e^{ - z \left[ X\right] }}\right]} \right){\rm{E}}\left[{e^{ - z\left[ Y \right]}}\right]{\rm{d}}z.
	\vspace{0mm}
\end{equation}
In (\ref{CommunicationBasicEquation}), ${\rm{E}}\left[{e^{ - z \left[ X\right] }}\right]$ and ${\rm{E}}\left[{e^{ - z \left[ Y\right] }}\right]$ are the Laplace transforms of $X$ and $Y$.
Then, exploiting the BSs for cooperative joint transmission within the range $D$, the expectation of data rate can be expressed as follows:
\vspace{0mm}
\begin{equation}\label{LapalaceTransform}
	\begin{aligned}
		&{\rm{E}}\left[ {\log \left( {1 + \mathrm{SIR}_c} \right)}  \right] 
		= {\rm{E}} \! \left[ {\log \left( \! {1 + \frac{\sum\nolimits_{{{i}} \in {\Phi_c}} {g_{i}} \left\| {\bf{d}}_i \right\|^{-\alpha} }{ \sum\nolimits_{{{j}} \in \{\Phi_b \backslash \Phi_c \}} g_j \left\| {\bf{d}}_j \right\|^{-\alpha}  }} \right)} \right] \\
		=& \int_0^\infty  {\frac{{1 - {\rm{E}}\left[ {{e^{ - z U}}} \right]}}{z}}  {\rm{E}}\left[ {{e^{ - z I}}} \right]{\rm{d}}z,
		\vspace{0mm}
	\end{aligned}
\end{equation}
where \( U = \sum\nolimits_{{{i}} \in {\Phi_c}} {g_{i}} \left\| {\bf{d}}_i \right\|^{-\alpha} \) and \( I =  \sum\nolimits_{{{j}} \in \{\Phi_b \backslash \Phi_c \}} g_j \left\| {\bf{d}}_j \right\|^{-\alpha} \). In (\ref{LapalaceTransform}), \( I \) represents the interference arising from the BSs located outside the cooperative region. The terms \( g_{i} \) denote the effective channel gains of the desired signal, where \( g_{i} \sim \Gamma \left( M_{\mathrm{t}} - 1, p^c \right) \), as described in \cite{Hosseini2016Stochastic}. According to the definition provided below equation (\ref{SIRexpression}), the distribution of \( g_{j} \) can be derived using the moment matching technique \cite{Hosseini2016Stochastic}.
Given that \( \mathrm{E}[p^c\left|\mathbf{h}_{j}^H \mathbf{w}_j^c\right|^2] = p^c \) and \( \mathrm{E}[p^s\left|\mathbf{h}_{j}^H \mathbf{w}_j^s\right|^2] = p^s \), we obtain \( \mathrm{E}[g_{j}] = p^s + p^c = 1 \). Moreover, since
$
\mathrm{E}[g_{j}^2] = \mathrm{E}\left[\left|\mathbf{h}_{j}^H \mathbf{w}_j^s\right|^4\right] + \mathrm{E}\left[\left|\mathbf{h}_{j}^H \mathbf{w}_j^c\right|^4\right] + 2\mathrm{E}\left[\left|\mathbf{h}_{j}^H \mathbf{w}_j^s\right|^2 \left|\mathbf{h}_{j}^H \mathbf{w}_j^c\right|^2\right] = (p^s + p^c)^2 = 1,
$, the interference channel gain \( g_{j} \) can be approximated by a gamma-distributed random variable. Consequently, \( g_{j} \sim \Gamma(1, 1) \).

Based on the above discussions, the useful signal power can be expressed by 
\vspace{0mm}
\begin{equation}\label{UsefulSignalPower}
	\begin{aligned}
		{\rm{E}}\left[ {{e^{ - zg_{1}}}} \right] \simeq \int_0^\infty  {\frac{{{e^{ - zx}}{x^{{M_{\mathrm{t}} - 2 }}}{e^{-\frac{x}{p^c}}}}}{{(p^c)^{M_{\mathrm{t}} - 1}\Gamma \left( {M_{\mathrm{t}} - 1} \right)}}} {\rm{d}}x = {\left( {1 + p^cz} \right)^{1 - {M_{\mathrm{t}} }}}.
		\vspace{0mm}
	\end{aligned}
\end{equation}
Then, we derive tight bounds on the Laplace transform of the cooperative transmission power and on the communication interference as follows.

\begin{thm}\label{LaplaceTransform}
	The Laplace transforms of $U$ and $I$ are given by
	\vspace{0mm}
	\begin{equation}\label{LaplaceTransformU}
		{\rm{E}}\left[ {{e^{ - z U}}} \right] = \exp \left[ - \pi \lambda_b {\rm{H}}_1\left( { zp^c,M_{\mathrm{t}}-1,\alpha, D } \right) \right],
	\end{equation}
	\begin{equation}\label{LaplaceTransformI}
		{\rm{E}}\left[ {{e^{ - zI}}} \right] =   \exp \left[  - \pi  \lambda_b {\rm{H}}_2\left( {z, \alpha , D }  \right) \right],
		\vspace{0mm}
	\end{equation}
	where ${\rm{H}}_1\left( {x,K,\alpha , D } \right)  = K{x^{\frac{2}{\alpha }}}\!\left(\! {{\bar B}\left( {\frac{x}{{x + D ^{\alpha}}},1  - \frac{2}{\alpha },K + \frac{2}{\alpha }} \right) } \right) + D^2\left( {1 - {{{{\left( {1 + x{D ^{-\alpha} }} \right)}^{-K}}}}} \right) $ and ${\rm{H}}_2\left( {x, \alpha , D} \right) =  D^2\left( {{{{{\left( {1 + x{D ^{-\alpha} }} \right)}^{-1}}}}} -1 \right) + {x^{\frac{2}{\alpha }}}\!\left(\! {{ B}\left( {\frac{x}{{x + D ^{\alpha}}},1  - \frac{2}{\alpha },1 + \frac{2}{\alpha }} \right) } \right)$.
\end{thm}
\begin{proof}
	Please refer to Appendix D.
\end{proof}

Based on the Laplace transforms of $U$ and $I$ in (\ref{LaplaceTransformU}) and (\ref{LaplaceTransformI}), the expected data rate is formulated in Theorem \ref{CommunicationTightExpression}.

\begin{theorem}\label{CommunicationTightExpression}
	The communication performance is characterized by 
	\vspace{0mm}
	\begin{equation}\label{TightCommunicationExpression}
		\begin{aligned}
			R_c =&\int_0^\infty  {\frac{{1 - \exp \left[ - \pi \lambda_b {\rm{H}}_1\left( { zp^c,M_{\mathrm{t}}-1,\alpha, D } \right) \right]}}{z}}  \\
			&\times \exp \left[  - \pi  \lambda_b {\rm{H}}_2\left( {z, \alpha , D }  \right) \right]{\rm{d}}z,
			\vspace{0mm}
		\end{aligned}
	\end{equation}
where $\lambda_b = \lambda_t / M_{\mathrm{t}}$.
\end{theorem}
\begin{proof}
	According to (\ref{CommunicationBasicEquation}), by substituting the Laplace transforms of useful signal and interference in Lemma \ref{LaplaceTransform} into (\ref{LapalaceTransform}), the conditional expected spectrum efficiency is given by
	%\vspace{0mm}
	\begin{equation}\label{longEquation2}
		\begin{aligned}
			&\int_0^\infty \int_0^\infty  {\frac{{1 - \exp \left[ - \pi \lambda_b {\rm{H}}_1\left( { zp^c,M_{\mathrm{t}}-1,\alpha, D } \right) \right]}}{z}}  \\
			&\times \exp \left[  - \pi  \lambda_b {\rm{H}}_2\left( {z, \alpha , D }  \right) \right] {f_r}\left( r \right) {\rm{d}}r {\rm{d}}z ,
			\vspace{0mm}
		\end{aligned}
	\end{equation}
	where we have ${f_r}\left( r \right) = 2\pi {\lambda _b}r{e^{ - \pi {\lambda _b}{r^2}}}$. Here, we note that since the collaboration region is fixed with a constant radius \(D\), the effective signal and interference contributions are determined solely by the distance \(r\) to the nearest BS, eliminating the need to consider the joint distance distribution of BSs within the cooperation region.
	Then, by solving the integral with respect to $r$, (\ref{TightCommunicationExpression}) can be obtained. 
	This completes the proof.
\end{proof}

According to (\ref{TightCommunicationExpression}), the communication rate increases monotonically with the increase of $D$, which is due to having more BSs on average participating in cooperative transmission, while users receive less interference. We will show in Section \ref{SimulationsSection} that the tractable expression given in (\ref{TightCommunicationExpression}) is closely approximated by Monte Carlo simulations.  

\subsection{Optimal Antenna-to-BS Allocation Analysis}
\label{ApproximatedCommunication}

To analyze the optimal antenna-to-BS allocation, we adopt simplifications for maximizing the expected SIR. First, we simplify the expected data rate as
\vspace{0mm}
\begin{equation}\label{ApproximationSIR}
%	\begin{aligned}
		{{\rm{E}}_{r,\Phi _b^S,g_i}}\left[ {\ln \left( {1 + \frac{S}{I}} \right)} \right] %&\le {{\rm{E}}_r}\left[ {\ln \left( {1 + {{\rm{E}}_{\Phi _b,g_i}}\left[ {\frac{S}{I}} \right]} \right)} \right] \\
		\approx {{\rm{E}}_r}\left[ {\ln \left( {1 + \frac{{{{\bar S}_r}}}{{{{\bar I}_r}}}} \right)} \right],
		\vspace{0mm}
%	\end{aligned}
\end{equation}
where ${\bar S}_r = {\rm{E}}_{\Phi _b, g_i} [ {g_{1}} + \sum\nolimits_{{{i}} \in {\Phi_a}} {g_{i}} \left\| {\bf{d}}_i \right\|^{-\alpha} r^{\alpha} ] = p^c {r^{ - \alpha }} + \frac{{\pi {\lambda _b}}}{{\alpha  - 2}}\left( {{r^{ - \alpha  + 2}} - D^{ - \alpha  + 2}} \right)$ and ${\bar I}_r = {\rm{E}}_{\Phi _b, g_i} \left[ \sum\nolimits_{{{j}} \in \{\Phi_b \backslash \Phi_a\}} g_j \left\| {\bf{d}}_j \right\|^{-\alpha} r^\alpha  \right] = \frac{{\pi {\lambda _b}}}{{\alpha  - 2}}{D^{2 - \alpha }}$. Based on (\ref{ApproximationSIR}), we simplify the antenna allocation optimization by reformulating the objective from maximizing the expected spectral efficiency to maximizing the expected SIR, thereby streamlining the analysis. Then, it follows that
\vspace{0mm}
\begin{equation}\label{SimplifiedCommunicationExpression}
\begin{aligned}
	{{\rm{E}}_{r,\Phi _b^S,g_i}}\left[ {\ln \left( {1 + \frac{S}{I}} \right)} \right] 
	\approx  {{\rm{E}}_r}\left[ {\ln \left( {1 + {\overline {\rm{SIR}}} } \right)} \right],
\end{aligned}
	\vspace{0mm}
\end{equation}
where we have $\overline {\rm{SIR}} = M_{\mathrm{t}}\left( {\left({{\frac{{\alpha  - 2}}{{\pi {\lambda _b}}}{r^{ - \alpha }} + {r^{ - \alpha  + 2}}}}\right)D^{ \alpha  - 2} - 1} \right).$
In (\ref{SimplifiedCommunicationExpression}), we consider the optimal communication design, where the beamforming gain is \( M_{\mathrm{t}} \) rather than \( p^c(M_{\mathrm{t}} - 1) \).
The $\overline {\rm{SIR}}$ value increases monotonically with the size of cooperative regions $D$. Due to the complicated integral operation of the distance distribution $r$ in (\ref{SimplifiedCommunicationExpression}), it is challenging to directly obtain the optimal antenna-to-BS allocation strategy based on (\ref{SimplifiedCommunicationExpression}). Thus, we analyze the relationship between the optimal number of antennas $M_{\rm{t}}$ and the expected SIR of the typical communication user. In the following, we use an approximate method for analysis, where a sufficiently small value of $\epsilon$ is chosen as the lower limit of integration to ensure the feasibility of the integral and the validity of the analysis.
\vspace{0mm}
	\begin{equation}\label{ExpectedSIR}
	\begin{aligned}
		&{\rm{E}}_r\left[ {\frac{{{{\bar S}_r}}}{{{{\bar I}_r}}}} \right] = \int_\epsilon^D  {M_{\mathrm{t}}\left( {\frac{{\frac{{\alpha  - 2}}{{\pi {\lambda _b}}}{r^{ - \alpha }} + {r^{ - \alpha  + 2}}}}{{{D^{ - \alpha  + 2}}}} - 1} \right)} f_r(r) dr \\
		=& M_{\mathrm{t}} \bigg( \frac{{\alpha  - 2}}{{{D^{ - \alpha  + 2}}}}{{\left( {\pi \lambda_b } \right)}^{\frac{{\alpha  - 2}}{2}}}\int_\epsilon^{\pi \lambda_b {D^2}} {{u^{ - \frac{\alpha }{2}}}} {e^{ - u}}du  \\
		&	\!+ \!\frac{{{{\left( {\pi \lambda_b } \right)}^{\frac{{\alpha  - 2}}{2}}}}}{{{D^{ - \alpha  + 2}}}}\int_\epsilon^{\pi \lambda_b {D^2}} \!\! {{u^{\frac{{ - \alpha  + 2}}{2}}}{e^{ - u}}du}  + {e^{ - \pi \lambda_b {D^2}}} \!- \!1 \!\bigg) 
		\triangleq  G(M_{\mathrm{t}}).
		\vspace{0mm}
	\end{aligned}
\end{equation}
Building on the above analysis, in the following, we investigate the optimal antenna-to-BS allocation in two specific cases.

\begin{Pro}\label{PropositionAlpha2}
	When $\alpha \to 2$, the optimal antenna-to-BS allocation strategy for communication is centralized, i.e., $\lambda_b^* = \frac{1}{\pi D^2}$ and $M_{\rm{t}}^* = \lambda_t \pi D^2$.
\end{Pro}
\begin{proof}
	Let ${\int_\epsilon^{\pi \lambda {D^2}} {{u^{ - \frac{\alpha }{2}}}} {e^{ - u}}du} = \vartheta$. Upon substituting it into (\ref{ExpectedSIR}), due to $\lim\limits_{\epsilon \to 0}\int_\epsilon^{\pi \lambda_b {D^2}} {{u^{\frac{{ - \alpha  + 2}}{2}}}{e^{ - u}}du} = \gamma \left( {1,\pi \lambda_b {D^2}} \right)$, it follows that
	\vspace{0mm}
	\begin{equation}
			{\rm{E}}\left[ \frac{{\bar{S}_r}}{{\bar{I}_r}} \right] = M_{\mathrm{t}}\left( \vartheta +  {\gamma \left( {1,\pi \lambda_b {D^2}} \right) + {e^{ - \pi \lambda_b {D^2}}} - 1} \right).
			\vspace{0mm}
	\end{equation}
	Due to $\gamma(1, x) = 1 - e^{-x}$, ${\rm{E}}\left[ \frac{{\bar{S}_r}}{{\bar{I}_r}} \right] = M_{\mathrm{t}} \vartheta$, we have $\vartheta  > 0$ since ${\rm{E}}\left[ \frac{{\bar{S}_r}}{{\bar{I}_r}} \right] \ge 0$. Therefore, the derivative of $G(M_{\mathrm{t}})$ follows 
	\vspace{0mm}
	\begin{equation}
		G'(M_{\mathrm{t}}) =  \frac{{\vartheta}}{2} \ge 0.
		\vspace{0mm}
	\end{equation}
	Thus, given a region $|{\cal{A}}|$, the number of antennas is $|{\cal{A}}| \times \lambda_t$, and the expected SIR of the communication user increases upon increasing the number of antennas allocated at each location. Therefore, the optimal antenna-to-BS allocation strategy is centralized, i.e., $\lambda_b^* = \frac{1}{\pi D^2}$ and $M_{\rm{t}}^* = \lambda_t \pi D^2$.
\end{proof}

\begin{Pro}\label{PropositionAlphaLarge}
	When $\alpha \gg 4$, the optimal antenna-to-BS allocation strategy for communication is distributed, i.e., $\lambda^*_b = \lambda_t$, $M_{\mathrm{t}}^* = 1$.
\end{Pro}
\begin{proof}
	In (\ref{ExpectedSIR}), we have $\int_\epsilon ^{\pi \lambda {D^2}} {{u^{ - \frac{\alpha }{2}}}} {e^{-u}}du \approx \frac{{2{\epsilon ^{ - \frac{\alpha }{2} + 1}}}}{{\alpha  - 2}}$ and $\int_\epsilon ^{\pi \lambda {D^2}} {{u^{\frac{{ - \alpha  + 2}}{2}}}{e^{ - u}}du}  \approx \frac{{2{\epsilon ^{\frac{{ - \alpha  + 4}}{2}}}}}{{\alpha  - 4}}$. Then, it follows that
	\vspace{0mm}
	\begin{equation}\label{G_CMexpression}
		G(M_{\mathrm{t}}) = M_{\mathrm{t}}\left( {C_0{{M_{\mathrm{t}}}^{\frac{{2 - \alpha }}{2}}} + {e^{ - \frac{\pi \lambda_t {D^2}}{M_{\rm{t}}}}} - 1} \right),
		\vspace{0mm}
	\end{equation}
where $C_0 = \frac{2}{{{D^{ - \alpha  + 2}}}}{\left( {\pi {\lambda _t}} \right)^{\frac{{\alpha  - 2}}{2}}}{\epsilon ^{ - \frac{\alpha }{2} + 1}}\left( {1 + \frac{\epsilon }{{\alpha  - 4}}} \right)$. Due to $\int_\epsilon ^{\frac{\pi \lambda_t {D^2}}{M_{\rm{t}}}} {{u^{ - \frac{\alpha }{2}}}} {e^{ - u}}du \ge \frac{{2{\epsilon ^{ - \frac{\alpha }{2} + 1}}}}{{\alpha  - 2}}$, the derivative of $G(M_{\mathrm{t}})$ becomes as follows 
\vspace{0mm}
	\begin{equation}\label{InfityAlpha}
		\begin{aligned}
			G'(M_{\mathrm{t}}) =& \left( {2{{\left( {\frac{{\pi {\lambda _t}{D^2}}}{M_{\mathrm{t}}}} \right)}^{\frac{{\alpha  - 4}}{2}}}{\epsilon ^{ - \frac{\alpha }{2} + 1}}\left( {1 + \frac{\epsilon }{{\alpha  - 4}}} \right)\frac{{4 - \alpha }}{2} + 1} \right) \\
			& \times \frac{{\pi {\lambda _t}{D^2}}}{M_{\mathrm{t}}} \le 0.
			\vspace{0mm}
		\end{aligned}
	\end{equation}
The inequality in (\ref{InfityAlpha}) holds due to ${2{{\left( {\frac{{\pi {\lambda _t}{D^2}}}{M_{\mathrm{t}}}} \right)}^{\frac{{\alpha  - 4}}{2}}}{\epsilon ^{ - \frac{\alpha }{2} + 1}}\left( {1 + \frac{\epsilon }{{\alpha  - 4}}} \right)\frac{{4 - \alpha }}{2}} \ll -1$ when $\alpha \gg 4$.
Therefore, $G(M_{\mathrm{t}})$ increases, as $M_{\mathrm{t}}$ decreases.
\end{proof}

The above conclusions provide practical guidance for antenna-to-BS allocation. In environments suffering from strong fading, a distributed antenna allocation strategy should be adopted for positioning service antennas closer to the users. In such scenarios, the distributed antenna allocation enhances the useful signal, because although the interference may be increased, the resulting mitigation of fading is typically more substantial. Conversely, in environments having mild fading, such as line-of-sight-dominant channels, distributed allocation may significantly amplify the interference and this is not outweighed by the fading mitigation. In these cases, a centralized allocation can be more effective, as beamforming techniques can be used to strengthen the desired signal, while reducing interference.

\section{System Optimization}
In this section, we analyze the optimization of cooperative ISAC networks to demonstrate that antenna-to-BS allocation introduces a new degree of freedom, enabling a flexible balance between sensing and communication performance. Based on the derivations in Sections \ref{SensingSection} and \ref{CommunicationSection}, the derived tractable performance expressions of both S\&C are functions of the number of antennas and BS density. Then, we present a performance metric, namely the rate-CRLB performance region, to  verify the effectiveness of the proposed cooperative ISAC schemes. Without loss of generality, the S-C network performance region is defined as
\vspace{0mm}
\begin{equation}\label{PerformanceBoundry}
	\begin{aligned}
		{\cal{C}}_{c-s}(L,N, p^c, p^s)  = &\bigg\{ ( \hat r_c, {\hat {\rm{crlb}}} ): \hat r_c \le R_c, {\hat {\rm{crlb}}} \ge {\rm{CRLB}}, \\
		& p^s + p^c \le 1, M_{\rm{t}} \lambda_b \le \lambda_t, M_{\rm{r}} \lambda_b \le \lambda_r \bigg\},
				 \vspace{0mm}
	\end{aligned}
\end{equation}
where $(\hat r_c, \hat {\rm{crlb}})$ represents an achievable rate-CRLB performance pair. By examining this region, we gain a clear view of how improvements in one domain (e.g., increasing communication rate) may impact the other (e.g., localization accuracy) and vice versa. In this case, the rate-CRLB performance region can be utilized to characterize all the achievable communication rate and achievable CRLB pairs under the constraint of the antenna resources.
%A direct way to find the boundary of the rate-CRLB region is to exhaustively search through the entire of set all feasible variables $(M_{\rm{t}}, \lambda_t, p^c, p^s)$ and calculate the corresponding S\&C performance expressions derived in Sections \ref{SensingSection} and \ref{CommunicationSection}. However, this operation imposes an excessive computational complexity, especially when there are abundant antenna resource. 

It is not difficult to verify that $R_c$ is monotonically increasing with the communication transmit power $p^c$, while the $\rm{CRLB}$ is monotonically decreasing with the sensing transmit power $p^s$. Therefore, under a certain power allocation, if the sensing and communication performance $(\hat r_c^*, \hat {\rm{crlb}}^*)$ at the current BS density is better than the performance under all other BS density configurations $(\hat r_c', \hat {\rm{crlb}}')$, then the power allocations at the latter BS densities are definitely not on the performance boundary. 
Indeed, it is sufficient to explore the two dimensions of $\lambda_b$ and $p^c$ individually, instead of using a two-dimensional exhaustive search.
Then, according to the optimal cooperative BS density of communication-only and sensing-only networks, denoted by $\lambda_b^*(c)$ and $\lambda_b^*(s)$, the search range can be drastically reduced within $[\min(\lambda_b^*(c), \lambda_b^*(s)), \max(\lambda_b^*(c), \lambda_b^*(s))]$.

\section{Simulations}
\label{SimulationsSection}

Using numerical results, we have studied the fundamental characteristics of ISAC networks and verified the tightness of the derived tractable expression by comparing it to Monte Carlo simulation results in this section. Our numerical simulations are averaged over network topologies and small-scale channel fading realizations. The system parameters are given as follows: the number of transmit antennas $M_{\mathrm{t}} = 4$; the number of receive antennas $M_{\mathrm{r}} = 10$; the transmit power $P_{\mathrm{t}} = 1$W at each BS; $|\xi|^2 = 1$; the transmit and receive antenna density $\lambda_t = \lambda_r = 50/km^2$; the frequency $f^c = 5$ GHz; the bandwidth $B \in [10, 100]$ MHz; the noise power $-100$dB; pathloss coefficients $\alpha = 4$ and $\beta = 2$.

\begin{figure}[t]
	\centering
	\includegraphics[width=7.7cm]{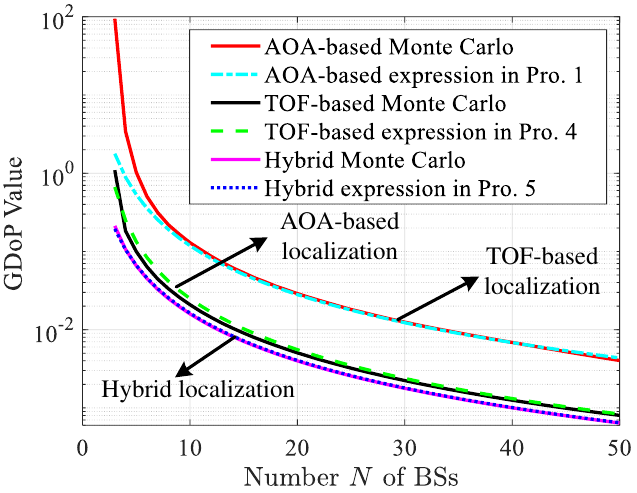}
	\vspace{0mm}
	\caption{GDoP versus BS number with three different localization methods: AOA-based, TOF-based, and Hybrid localization.}
	\label{figure4}
\end{figure}

To verify the accuracy of our sensing performance analysis, our Monte Carlo simulations are compared to the closed-form expression derived in Section \ref{SensingSection}, as shown in Fig.~\ref{figure4}. Specifically, the disparity between the simulation results and the results outlined in Propositions \ref{GDoPDerivation}, \ref{GDoPDerivationRangeBased}, and \ref{GDoPDerivationHybrid} is remarkably small.
This demonstrates the effectiveness of the GDoP expression presented in Propositions \ref{GDoPDerivation}, \ref{GDoPDerivationRangeBased}, and \ref{GDoPDerivationHybrid}.
It can be seen from Fig. \ref{figure4} that the TOF-based localization mechanism offers greater geometric gains, as it leverages directional diversity in both the transmission and reception stages, whereas AOA-based localization only utilizes geometric diversity at the receiving end. 
Typically, when the size of the cooperative sensing cluster increases from $N = 3$ to $N = 6$, the geometry-based gain of these three localization methods increases tenfold.
When the number of BSs is $N=2$, the hybrid localization method significantly reduces the GDoP values compared to TOF-based and AOA-based localization methods, namely by factors of 5 and 50 respectively, because it can fuse two types of measurement information for significantly increasing the amount of geometric information, when the number of BSs is small.

\begin{figure}[t]
	\centering
	\includegraphics[width=7.7cm]{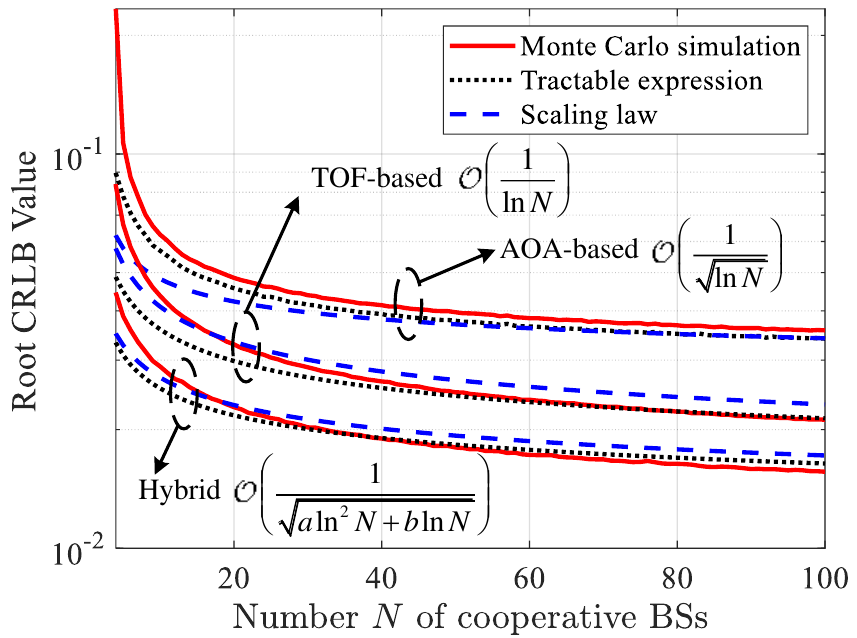}
	\vspace{0mm}
	\caption{Localization performance scaling law with respect to the cooperative BS number $N$ under the fixed number of antennas per BS.}
	\label{figure5}
\end{figure}
%The area length is $1 km^2$, 
In Fig.~\ref{figure5}, given $M_t = 4$ and $M_r = 10$, and a bandwidth $B = 10$ MHz, the scaling law of the CRLB expressions derived in Theorems \ref{SimplifiedWithDis3}, \ref{SimplifiedWithDisRange}, and Proposition \ref{ScalingLawHybrid} is also consistent with the simulation results. 
It is important to note that when the number of cooperating BSs is relatively small, specifically when $N \le 4$, the closed-form expressions show a slight deviation from the results obtained through Monte Carlo simulations. This discrepancy arises primarily from the reduced accuracy in calculating the expectation of trigonometric functions, such as ${\rm{E}}[\sin(\cdot)]$ and ${\rm{E}}[\cos(\cdot)]$, when the number of ISAC BSs is limited.
When the number of BSs is small, the hybrid localization method can significantly improve performance. This is primarily because the geometric arrangement of the BSs relative to the target may be suboptimal, leading to poor performance in localization methods that rely solely on ranging or AOA measurements.
Fig. \ref{figure5} shows that increasing the number of cooperative BSs significantly improves accuracy, when the total number of BSs is limited. However, the performance gains become marginal once $N \ge 10$. This is expected, as adding more distant and randomly located BSs leads to increased signal attenuation, offering diminishing returns compared to nearby BSs. As depicted in Fig. \ref{figure5}, the expected CRLB decreases more slowly upon increasing $N$ than the GDoP value. Additionally, hybrid localization, which combines TOF and AOA estimation results, can greatly enhance accuracy, when the number of BSs is small.

\begin{figure}[t]
	\centering
	\includegraphics[width=7.7cm]{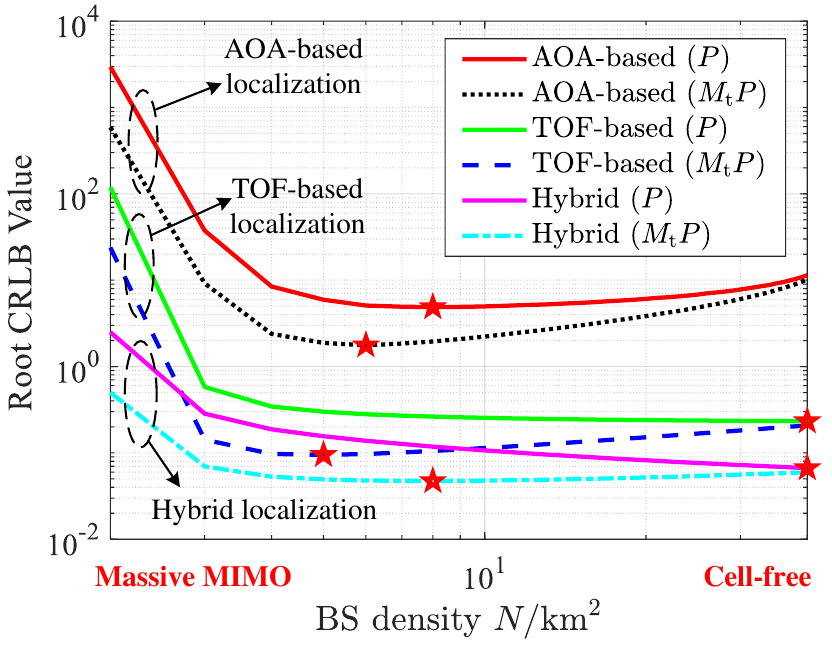}
	\vspace{0mm}
	\caption{Localization performance comparisons with respect to the cooperative BS density under the fixed antenna density ($P$ and $M_t \cdot P$ refer to the total power constraint on each BS).}
	\label{figure6}
\end{figure}

In Fig. \ref{figure6}, both the transmit and receive antenna densities are set to $\lambda_t = \lambda_r = 50/\text{km}^2$. The noise power is $\sigma_s^2 = -100$ dB, and the bandwidth is $B = 10$ MHz. In the legend, \(P\) denotes the fixed energy allocated under a station-level fixed power constraint, ensuring that each BS maintains constant total power regardless of the number of antennas. Conversely, \(M_{\rm{t}} \cdot P\) indicates that the BS’s total power scales proportionally with the number of antennas \(M_{\rm{t}}\), thereby keeping the overall network power constant within the cooperation region. Consistent with our analysis, Fig. \ref{figure6} shows that under $P$ power constraint, the optimal allocation strategy for the TOF-based and hybrid localization methods is a fully distributed configuration. By contrast, for AOA-based localization, the optimal allocation requires concentrating a certain number of antennas to improve the AOA estimation accuracy, resulting in an ideal allocation of eight BSs per square kilometer. 
When $M_t \cdot P$ power constraints are applied, both TOF-based and hybrid methods tend to favor a mixed configuration that combines centralized and distributed allocation. This approach better balances the beamforming gain of multiple antennas with the macro-diversity gain. Under optimal allocation conditions, our proposed hybrid localization method reduces the localization error to just 1.3\% and 28.8\% of that achieved by AOA-based and TOF-based localization methods, respectively, significantly enhancing localization accuracy.

\begin{figure}[t]
	\centering
	\includegraphics[width=7.7cm]{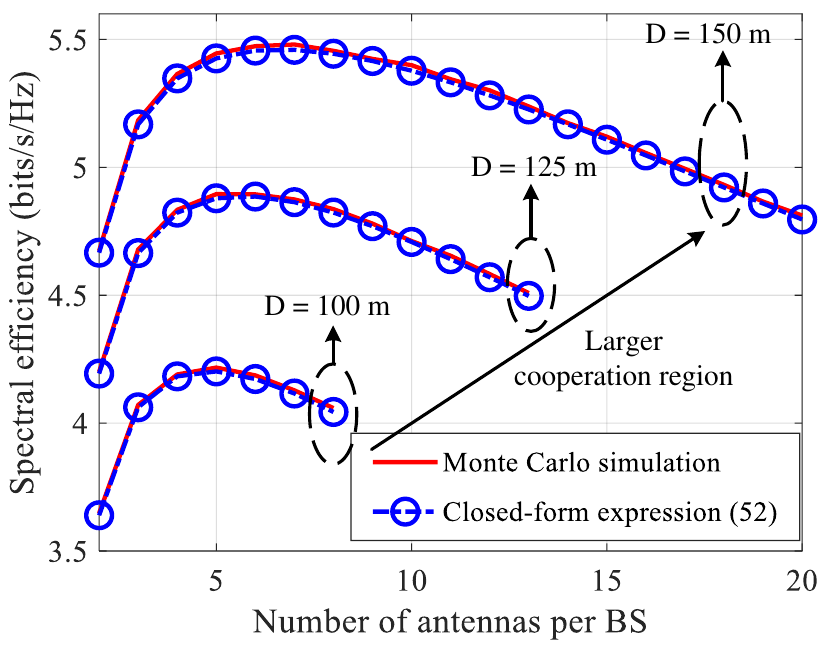}
	\vspace{0mm}
	\caption{Antenna allocation versus spectrum efficiency with different cooperative ranges $D$ = 100, 125, and 150 m.}
	\label{figure8}
\end{figure}

Fig. \ref{figure8} shows that our tractable expression derived for the communication rate closely aligns with the Monte Carlo simulations, given an antenna density of $\lambda_t = 300/\text{km}^2$. 
Fig. \ref{figure8} also shows that the spectral efficiency \(R_c\) initially increases with the number of antennas per BS but then decreases. This is because the initial improvement in communication performance attained by the beamforming gain is eventually outweighed by the performance erosion resulting from the increased average serving distance, which is due to the reduced BS density.
As the radius \(D\) of the cooperative area expands, the optimal communication rate increases, mainly due to the higher signal power and reduced interference power. Additionally, for a larger cooperative area, the optimal number of antennas per BS also increases to maximize communication rates. This is because a larger area provides more antenna resources, and adding antennas at each BS improves beamforming gain, which helps mitigate the path loss associated with the expanded cooperative area.

\begin{figure}[t]
	\centering
	\includegraphics[width=7.7cm]{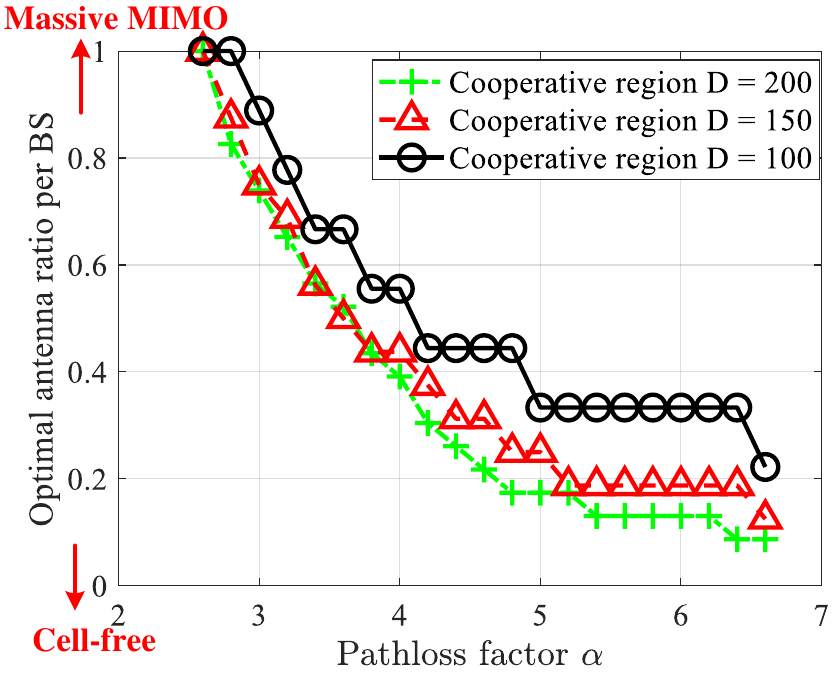}
	\vspace{0mm}
	\caption{Antenna allocation versus path loss factor $\alpha$ with various radius $D$ = 100, 150, 200 m.}
	\label{figure9}
\end{figure}

In Fig. \ref{figure9}, it is observed that as the attenuation coefficient \( \alpha \) increases, a distributed allocation becomes more favorable, which is consistent with our analysis in Propositions \ref{PropositionAlpha2} and \ref{PropositionAlphaLarge}. This is because distributing antennas reduce the average distance between the service BS and users, making it possible to ignore interference from distant sources, as both the effective signals and interference from far-off locations become negligible.
Conversely, when the attenuation coefficient \( \alpha \) decreases, a centralized allocation becomes more advantageous. This is because interference from distant BSs has a greater impact. Even though centralized allocation increases the average distance between the service BS and users, the reduced attenuation coefficient ensures that signals still reach the users with sufficient strength. It is observed from Fig. \ref{figure9} that as the cooperative area radius \(D\) increases, the optimal proportion of antennas deployed at each BS relative to the total number of antennas in the area gradually decreases, for example, from 55\% at \(D=100\) to 39\% at \(D=200\). The primary reason is that the expansion of the cooperative area involves more antennas in cooperation, requiring a more dense distribution of antennas closer to users to enhance communication rates.

%\subsection{Sensing and Communication Performance Tradeoff}

\begin{figure}[t]
	\centering
	\includegraphics[width=7.7cm]{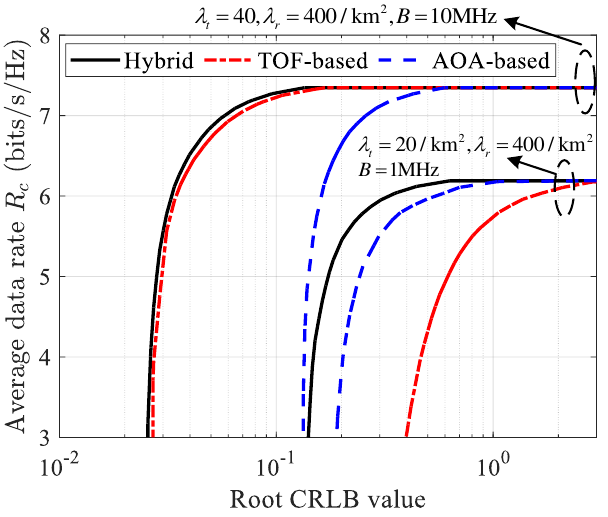}
	\vspace{0mm}
	\caption{Rate-CRLB performance boundary with different localization methods.}
	\label{figure10}
\end{figure}

Fig. \ref{figure10} illustrates the performance boundaries defined in (\ref{PerformanceBoundry}) for the optimal joint allocation and power allocation across three different localization schemes, of $D = 1000$ m. Notably, the hybrid localization scheme that combines both TOF and AOA information shows a significant enhancement in performance boundaries. Specifically, under ${\lambda _t} = 20/{\rm{k}}{{\rm{m}}^2}, {\lambda _r} = 400/{\rm{k}}{{\rm{m}}^2}$ and $B = 1{\rm{MHz}}$, the hybrid localization method reduces errors by a factor of 4.6 and 1.7 compared to TOF-based and AOA-based localization methods, respectively, while maintaining a communication spectral efficiency of 6 bits/Hz/s.
The overall S\&C performance boundaries improve as the transmit antenna density $\lambda_t$ increases. This is primarily due to the increased flexibility in resource allocation, which enables higher beamforming gains. As the bandwidth decreases, it can be observed from Fig. \ref{figure10} that AOA-based localization outperforms TOF-based localization. This is because the TOF measurement accuracy declines upon decreasing the bandwidth.

%Additionally, the boundary for the communication-sensing performance of the hybrid localization scheme is closer to a square, suggesting that this approach more effectively integrates the network architecture for both sensing and communication. Consequently, it optimizes the overall performance of the network in terms of both communication and sensing capabilities.

\section{Conclusions}
This work proposed an innovative cooperative ISAC network that combines multi-static sensing with CoMP data transmission, incorporating advanced localization methods that exploit both AOA and TOF measurements. Our study demonstrates that optimal antenna-to-BS allocation, through a balance of centralized and distributed configurations, significantly enhances the network performance by maximizing spatial diversity and coherent processing gains. Additionally, we provided analytical insights into the scaling laws of different localization techniques and establish a comprehensive framework for evaluating communication data rates across various antenna-to-BS allocation strategies. Our findings underscore the substantial benefits of hybrid localization approaches and cooperative resource allocation, offering a more flexible trade-off between sensing and communication performance. This work not only deepens the understanding of ISAC network dynamics but also lays a foundation for future designs that can better meet the dual demands of sensing and communication in complex wireless environments. The joint scheduling of antenna resources for multiple users and targets constitutes worthwhile future research. Additionally, network resource allocation schemes that incorporate target scanning or detection as well as antenna deployment costs will be investigated to optimize performance trade-offs, thereby enhancing both the theoretical framework and practical applicability of cooperative ISAC network designs.

\section*{Appendix A: \textsc{Proof of Proposition \ref{GDoPDerivation}}}
For ease of analysis, (\ref{GDoPExpression}) can be transformed as follows:
%\vspace{0mm}
\begin{equation}\label{GDoPexpressionTransformation}
	\begin{aligned}
		& {{\rm{tr}}\left( {\tilde {\bf{F}}}_{\mathrm{A}}^{-1} \right)} 
		\!=\! \frac{{N\sum\nolimits_{i = 1}^N {{{\cos }^2}{\theta _i}} }}{{{N^2}\left( {\sum\nolimits_{i = 1}^N {a_{i}^2} } \right)\left( {\sum\nolimits_{i = 1}^N {d^2_{i}} } \right) \!-\! {N^2}{{\left( {\sum\nolimits_{i = 1}^N {a_{i}f_{i}} } \right)}^2}}}\\
		&= \frac{{\sum\nolimits_{i = 1}^N {{{\cos }^2}{\theta _i}} }}{{N\left( {\sum\nolimits_{i = 1}^N {\sum\nolimits_{j > i}^N {{{\left( {\sin {\theta _i}\cos {\theta _i}{{\cos }^2}{\theta _j} - \sin {\theta _j}\cos {\theta _j}{{\cos }^2}{\theta _i}} \right)}^2}} } } \right)}} ,
%		\vspace{0mm}
	\end{aligned}	
\end{equation}
where the expectation of the term ${\left( {\sin {\theta _i}\cos {\theta _i}{{\cos }^2}{\theta _j} - \sin {\theta _j}\cos {\theta _j}{{\cos }^2}{\theta _i}} \right)^2}$ of the denominator in (\ref{GDoPexpressionTransformation}) can be simplified to
\vspace{0mm}
\begin{equation}\label{GDoPExpressionTransformDenominator}
	\begin{aligned}
		&{\left( {\sin {\theta _i}\cos {\theta _i}{{\cos }^2}{\theta _j} - \sin {\theta _j}\cos {\theta _j}{{\cos }^2}{\theta _i}} \right)^2}\\
		%=& {\left( {\frac{{\sin 2{\theta _i}}}{2}\frac{{\cos 2{\theta _j} + 1}}{2} - \frac{{\sin 2{\theta _j}}}{2}\frac{{\cos 2{\theta _i} + 1}}{2}} \right)^2}\\
		%=& \frac{1}{{16}}{\left( {\sin 2{\theta _i}\cos 2{\theta _j} - \sin 2{\theta _j}\cos 2{\theta _i} + \sin 2{\theta _i} - \sin 2{\theta _j}} \right)^2}\\
		=& \frac{1}{{16}}{\left( {\sin \left( {2{\theta _i} - 2{\theta _j}} \right) + \sin 2{\theta _i} - \sin 2{\theta _j}} \right)^2} 
		\approx \frac{3}{{32}}.
		\vspace{0mm}
	\end{aligned}
\end{equation}
By substituting (\ref{GDoPExpressionTransformDenominator}) into ${\rm{E}}_{\theta}\left[ {{\rm{tr}}\left( {\tilde {\bf{F}}}_{\mathrm{A}}^{-1} \right)} \right]$, we arrive at:
\vspace{0mm}
\begin{equation}
		{\rm{E}}_{\theta}\left[ {{\rm{tr}}\left( {\tilde {\bf{F}}}_{\mathrm{A}}^{-1} \right)} \right]
		\approx  \frac{{32}}{{3N\left( {N - 1} \right)}}.
		\vspace{0mm}
\end{equation}
This thus completes the proof.

\section*{Appendix B: \textsc{Proof of Proposition \ref{SimplifiedWithDis1}}}
To facilitate the analysis, we transform ${\rm{CRLB}}_{\rm{A}}$ of (\ref{CRLBexpression}) into
\vspace{0mm}
\begin{equation}
	\begin{aligned}
		& {\rm{E}}_{\theta, d} \left[ \frac{{\sum\nolimits_{i = 1}^N {\frac{{f_i^2}}{{d_i^4}}} }}{{\left( {\sum\nolimits_{j = 1}^N {\frac{1}{{d_j^2}}} } \right)\left( { {\sum\nolimits_{i = 1}^N {\frac{{e_i^2}}{{d_i^4}}} } {\sum\nolimits_{i = 1}^N {\frac{{f_i^2}}{{d_i^4}}} }  - {{\left( {\sum\nolimits_{i = 1}^N {\frac{1}{{d_i^4}}{e_i}{f_i}} } \right)}^2}} \right)}} \right] \\
		= & {\rm{E}}_{\theta, d} \left[ \frac{{\sum\nolimits_{i = 1}^N {\frac{{f_i^2}}{{d_i^4}}} }}{{\left( {\sum\nolimits_{j = 1}^N {\frac{1}{{d_j^2}}} } \right)\left( {\sum\nolimits_{i = 1}^N {\sum\nolimits_{i > j}^N {\frac{1}{{d_i^4d_j^4}}{{\left( {{X_{i,j}}} \right)}^2}} } } \right)}} \right],
		\vspace{0mm}
	\end{aligned}
\end{equation}
where $X_{i,j} = {\sin {\theta _i}\cos {\theta _i}{{\cos }^2}{\theta _j} - \sin {\theta _j}\cos {\theta _j}{{\cos }^2}{\theta _i}}$. When the number of nodes is large, the correlation between the numerator and the denominator becomes negligible, as the distances of different nodes are statistically independent under a PPP distribution. Moreover, when the number of nodes is large, the variability of the denominator is relatively small compared to its expected value.
Utilizing the approximation in (\ref{GDoPExpressionTransformDenominator}), 
the expected CRLB can be approximated as 
\vspace{0mm}
\begin{equation}\label{SimplifiedExpressionCRLBX}
	{\rm{CRLB}}_{\rm{A}} =	\frac{16{\sum\nolimits_{i = 1}^N {{{{{\rm{E}}[d_i]}^{-\beta-2}}}} }}{{3 {\sum\nolimits_{k = 1}^N {{{{{\rm{E}}[d_k]}^{-\beta}}}} }  {\sum\nolimits_{i = 1}^N {\sum\nolimits_{i > j}^N {{{{{\rm{E}}[d_i]}^{-\beta-2}{{\rm{E}}[d_j]}^{-\beta-2}}}} } } }}.
	\vspace{0mm}
\end{equation}
This thus completes the proof.

\section*{Appendix C: \textsc{Proof of Theorem \ref{SimplifiedWithDis3}}}
When $\beta = 2$, it follows that
\vspace{0mm}
\begin{equation}
	{\rm{CRLB}}_{\rm{A}} \! \approx \! \frac{32|\zeta_a |^{-2}{\sum\nolimits_{i = 1}^N {{i^{ - 2}}} }}{{{3}{{}}{\lambda_b ^3}{\pi ^3} {\sum\nolimits_{k = 1}^N {{k^{ - 1 }}} } \left( {{{\left( {\sum\nolimits_{i = 1}^N {{i^{ -  2}}} } \right)}^2}  \!- \! \sum\nolimits_{i = 1}^N {{i^{ -4}}} }  \!\right)}} \!.
	\vspace{0mm}
\end{equation}
We have $\mathop {\lim }\limits_{N \to \infty } \sum\nolimits_{n = 1}^N \frac{1}{n} \approx \ln N + \gamma  + \frac{1}{{2N}}$ and $\mathop {\lim }\limits_{N \to \infty } \sum\nolimits_{n = 1}^N \frac{1}{n^2} \approx \frac{\pi ^2}{6}$, where $\gamma \approx 0.577$ is Euler's constant.
%\vspace{0mm}
\begin{equation}\label{AngleApproximationEquation}
	\begin{aligned}
		{\rm{CRLB}}_{\rm{A}} 
		\approx &  \frac{{\sum\nolimits_{i = 1}^N {{i^{ - 2}}} }}{{\frac{3}{{32}}{\lambda_b ^3}{\pi ^3}\left( {\sum\nolimits_{j = 1}^N {{i^{ - 1}}} } \right)\left( {{{\left( {\sum\nolimits_{i = 1}^N {{i^{ - 2}}} } \right)}^2} - \sum\nolimits_{i = 1}^N {{i^{ - 4}}} } \right)}} \\
		%\approx & \frac{{\frac{{16{\pi ^2}}}{9}}}{{{\lambda ^3}{\pi ^3}\left( {\ln N + \gamma  + \frac{1}{{2N}}} \right)\left( {\frac{{{\pi ^4}}}{{36}} - \frac{{{\pi ^4}}}{{90}}} \right)}} \\
		\approx & \frac{{320}}{{3{\lambda_b ^3}{\pi ^5}\left( {\ln N + \gamma  + \frac{1}{{2N}}} \right)}}.
%		\vspace{0mm}
	\end{aligned}
\end{equation}
This thus completes the proof.

\section*{Appendix E: \textsc{Proof of Lemma \ref{LaplaceTransform}}}
For the Laplace transform of the interference coming from the BSs outside the cooperative region, we have
\vspace{0mm}
\begin{align}\label{CommunicationEquationExpression}
	{{\cal L}_{{I}}}(z)  = & {\rm{E}}_{\Phi_b, g_i}\big[ {\exp \big( { - z{ }\sum\nolimits_{i \in {\Phi_b}}   {{{\left\| {{{\bf{d}}_i}} \right\|}^{ - \alpha }}} {{| {{\bf{h}}_{i}^H{{\bf{W}}_i}} |}^2}} \big)} \! \big] \nonumber  \\
	\overset{(a)}{=}&  {\rm{E}}_{\Phi_b}\!\left[ \!\left( \prod _{{{{\bf{d}}_i}} \in \Phi_b \textbackslash {\cal{C}}(0,D)} {   {  {{{{\left( {1 + z{{\left\| {{{\bf{d}}_i}} \right\|}^{ - \alpha }}} \right)}^{-1}}}} } dx} \right) \bigg| D \right]  \nonumber \\
	\overset{(b)}{=}& \exp \left( { - 2\pi \lambda_b \int_{{D}}^\infty  {\left( {1 - {{{{\left( {1 + z{x^{ - \alpha }}} \right)}^{-1}}}}} \right)} xdx} \right)  \nonumber \\
	\overset{(c)}{=}& \exp \bigg(  - \left. {\pi \lambda_b y\left( {1 - {{{{\left( {1 + z{y^{ - \frac{\alpha }{2}}}} \right)}^{-1}}}}} \right)} \right|_{{D^2}}^\infty   \nonumber \\
	& - \pi \lambda_b \int_{{D^2}}^\infty  {\frac{\alpha }{2}z {y^{ - \frac{\alpha }{2}}}{{\left( {1 + z {y^{ - \frac{\alpha }{2}}}} \right)}^{ - 2}}} dy \bigg).
%	\vspace{0mm}
\end{align}
In (\ref{CommunicationEquationExpression}), ($a$) follows from the fact that the small-scale channel fading is independent of the BS locations and that the interference power imposed by each interfering BS at the typical user is distributed as $\Gamma(1,1)$. To derive ($b$), we harness the probability generating functional of a PPP with a density of $\lambda_b$. To elaborate, ($c$) comes from the variable $y = x^2$ and the distribution integral strategies.
Then, we have
\begin{equation}
	\begin{aligned}
		& \int_0^{z{r^{ - \alpha }}} \!\! {{{\left( {1 + u} \right)}^{ - K - 1}}} {u^{ - \frac{2}{\alpha }}}du  	\overset{(d)}{=}  \int_0^{\frac{{z{r^{ - \alpha }}}}{{1 + z{r^{ - \alpha }}}}} \!\! {{{\left( {1 - x} \right)}^{K - 1 + \frac{2}{\alpha }}}} {x^{ - \frac{2}{\alpha }}}dx \\
		& =  B\left( {\frac{{z{r^{ - \alpha }}}}{{1 + z{r^{ - \alpha }}}},1 - \frac{2}{\alpha },K + \frac{2}{\alpha }} \right),
	\end{aligned}
\end{equation}
where ($d$) follows from the distribution integral strategies, ${u = \frac{x}{{1 - x}}}$, and by setting $K=1$ since the interference power gain is approximately Gamma distributed as $\Gamma(1,1)$.
Similarly, the Laplace transform of useful signals can be expressed by 
%\vspace{0mm}
\begin{equation}
	{\rm{E}}\!\left[ {{e^{ - z U}}} \right] \!=\! \exp \!\bigg( \! - \pi \lambda_b {\rm{H}}_1\left( { z p^c,M_{\mathrm{t}}-1,\alpha , D } \right) \!\bigg),
%	\vspace{0mm}
\end{equation}
% 记得修改-双栏
where ${\rm{H}}_1\left( {x,K,\alpha ,D } \right)  = {D^2}\left( {1 - {{{{\left( {1 + z{D^{ - \alpha }}} \right)}^{-K}}}}} \right) + K{z^{\frac{2}{\alpha }}}\bar B\left( {\frac{z}{{z + {D^\alpha }}},1 - \frac{2}{\alpha },K + \frac{2}{\alpha }} \right)$.
This completes the proof.

\vspace{0mm}
\footnotesize  	
\bibliography{mybibfile}
\bibliographystyle{IEEEtran}

\end{document}